\documentclass[sigplan,screen]{acmart}
\copyrightyear{2024}
\acmYear{2024}
\setcopyright{rightsretained}
\acmConference[PPoPP '24]{The 29th ACM SIGPLAN Annual Symposium on Principles and Practice of Parallel Programming}{March 2--6, 2024}{Edinburgh, United Kingdom}
\acmBooktitle{The 29th ACM SIGPLAN Annual Symposium on Principles and Practice of Parallel Programming (PPoPP '24), March 2--6, 2024, Edinburgh, United Kingdom}
\acmDOI{10.1145/3627535.3638483}
\acmISBN{979-8-4007-0435-2/24/03}

\usepackage[compact]{titlesec}
\titlespacing*{\section}
{0pt}{.9ex plus .5ex minus .2ex}{.3ex plus .1ex}
\titlespacing*{\subsection}
{0pt}{.7ex plus .5ex minus .2ex}{.3ex plus .1ex}

\titleformat{\section}{\Large\bfseries}{\thesection}{1em}{}
\titleformat{\subsection}{\large\bfseries}{\thesubsection}{1em}{}
\titleformat{\subsubsection}[runin]{\bfseries}{\thesubsubsection}{1em}{}
\def\fullversion
\usepackage{etoolbox}
\newcommand{\ifconference}[1]{{{\ifx\fullversion\undefined{#1}\fi}\xspace}}
\newcommand{\iffullversion}[1]{{{\ifx\conference\undefined{#1}\fi}\xspace}}

\usepackage{graphicx}  
\usepackage{lipsum}  
\newcommand{\hide}[1]{} 
\usepackage{xspace}
\usepackage{textcomp}
\usepackage{comment} 
\usepackage{verbatim}

\ifx\nocomment\undefined
\newcommand{\yan}[1]{{\textcolor{violet}{Yan: #1}}}
\newcommand{\yihan}[1]{{\color{blue}{\bf Yihan:} #1}}
\newcommand{\xiaojun}[1]{{\color{cyan}{\bf Xiaojun:} #1}}
\newcommand{\laxman}[1]{{\textcolor{purple}{Laxman: #1}}}
\else
\newcommand{\yan}[1]{}
\newcommand{\yihan}[1]{}
\newcommand{\xiaojun}[1]{}
\newcommand{\laxman}[1]{}
\fi

\usepackage{url}

\newcommand{\defn}[1]{\emph{\textbf{#1}}} 
\newcommand{\emp}[1]{\emph{\textbf{#1}}} 
\newcommand{\fname}[1]{\textsf{#1}} 
\newcommand{\vname}[1]{\mathit{#1}} 
\newcommand{\mathfunc}[1]{\mathit{#1}}

\newtheorem{theorem}{Theorem}[section]
\newtheorem{lemma}[theorem]{Lemma}

\let \originalleft \left
\let\originalright\right
\renewcommand{\left}{\mathopen{}\mathclose\bgroup\originalleft}
\renewcommand{\right}{\aftergroup\egroup\originalright}

\usepackage{scalerel} 

\newtheoremstyle{exampstyle}
{.5em} 
{1em} 
{\it} 
{.5em} 
{\it \bfseries} 
{.} 
{.5em} 
{} 
\theoremstyle{exampstyle} 
\theoremstyle{exampstyle} 
\theoremstyle{exampstyle} 
\theoremstyle{exampstyle} 

\makeatletter
\renewenvironment{proof}[1][\proofname]{\par
\vspace{-1\topsep}
\pushQED{\qed}%
\normalfont
\topsep0.1em \partopsep0pt 
\trivlist
\item[\hskip\labelsep
      \itshape
  #1\@addpunct{.}]\ignorespaces
}{%
\popQED\endtrivlist\@endpefalse
\addvspace{0.3em plus 0.2em} 
}

\newcommand{\whp}[1]{\emph{whp}\xspace}

\DeclareMathOperator*{\polylog}{polylog}

\usepackage{pifont}

\newcommand{\modelop}[1]{\texttt{#1}}
\newcommand{\forkins}{\modelop{fork}}

\newcommand{\thread}{thread}

\usepackage[shortlabels]{enumitem}

\setlist{topsep=0.3em,itemsep=0.2em,parsep=0.1em,leftmargin=*}

\usepackage{float}
\usepackage[font={small},aboveskip=0.1em, belowskip=0em]{caption}

\usepackage[list=true,skip=0em]{subcaption}
\usepackage[rightcaption]{sidecap}

\usepackage{wrapfig}

\usepackage{array}
\newcolumntype{L}[1]{>{\raggedright\let\newline\\\arraybackslash\hspace{0pt}}m{#1}}
\newcolumntype{C}[1]{>{\centering\let\newline\\\arraybackslash\hspace{0pt}}m{#1}}
\newcolumntype{R}[1]{>{\raggedleft\let\newline\\\arraybackslash\hspace{0pt}}m{#1}}
\newcolumntype{B}{>{\bf \boldmath}l}
\newcolumntype{P}{>{\boldmath\begin{math}}l<{\end{math}}}
\usepackage{rotating}

\usepackage{booktabs} 
\usepackage{multicol,multirow}
\usepackage{longtable} 
\usepackage{supertabular} 
\usepackage{colortbl}
\usepackage{bigstrut}

\usepackage{titlesec}

\newcommand{\mysubsubsection}[1]{\underline{#1}.}
\titleformat{\subsubsection}[runin]
{\normalfont\normalsize\bfseries}{\thesubsubsection}{1em}{\mysubsubsection}

\newcommand{\myparagraph}[1]{\vspace{.1em}\noindent\emp{#1}\enspace}
\newcommand{\myfirstparagraph}[1]{\noindent\emp{#1}\enspace}

\usepackage[ruled,lined,linesnumbered,noend]{algorithm2e}
\usepackage[noend]{algpseudocode}

\makeatletter
\patchcmd{\@algocf@start}
  {-1.5em}
  {0pt}
  {}{}
\newcommand{\nosemic}{\renewcommand{\@endalgocfline}{\relax}}
\newcommand{\dosemic}{\renewcommand{\@endalgocfline}{\algocf@endline}}
\newcommand{\popline}{\Indm\dosemic}

\SetSideCommentLeft
\SetKwInput{notations}{Notations}
\SetKwInput{notes}{Notes}
\SetKwInput{maintains}{Maintains}

\SetKwProg{myfunc}{Function}{}{}
\SetKwFor{parForEach}{ParallelForEach}{do}{endfor}
\SetKwFor{Justrepeat}{Repeat}{}{}

\SetKw{MIN}{min}
\SetKw{MAX}{max}
\SetKw{OR}{or}
\SetKw{AND}{and}

\definecolor{commentgreen}{RGB}{0,128,0}
\SetCommentSty{mycommfont}

\hide{
\SetCommentSty{mycommfont}
}
\makeatletter
\renewcommand{\LinesNumbered}{%
  \setboolean{algocf@linesnumbered}{true}%
  \renewcommand{\algocf@linesnumbered}{\everypar={\nl}}}%

\let\oldnl\nl
\newcommand{\nonl}{\renewcommand{\nl}{\let\nl\oldnl}}
\makeatother
\newcommand{\stepcomment}[1]{\vspace{.03in}\nonl \popline\textbf{\textcolor{commentgreen}{\footnotesize \enspace\,\textit{\underline{#1:}}}}\\\pushline}

\usepackage{mdframed}
\definecolor{framelinecolor}{RGB}{68,114,196}
\mdfdefinestyle{mystyle}{linecolor=framelinecolor,innertopmargin=1pt,innerbottommargin=2pt,backgroundcolor=gray!20,skipabove=2pt,skipbelow=0pt}
\mdfdefinestyle{densestyle}{linecolor=framelinecolor,innertopmargin=0,innerbottommargin=0,leftmargin=0,rightmargin=0,backgroundcolor=gray!20}
\mdfdefinestyle{compactcode}{linecolor=framelinecolor,innertopmargin=1pt,innerbottommargin=1pt,backgroundcolor=gray!20,skipabove=0pt,skipbelow=0pt,leftmargin=0,rightmargin=0}

\usepackage{framed}

\usepackage{listings}

\newdimen\zzsize
\newdimen\kwsize
\newcommand{\basicstyle}{\fontsize{\zzsize}{1\zzsize}\ttfamily}
\newcommand{\keywordstyle}{\fontsize{\kwsize}{1\kwsize}\ttfamily\bf}

\newdimen\zzlstwidth
\usepackage{cleveref}
\crefname{section}{Sec.}{Sec.}
\crefname{theorem}{Thm.}{Thm.}
\crefname{lemma}{Lem.}{Lem.}
\crefname{corollary}{Col.}{Col.}
\crefname{table}{Tab.}{Tab.}
\crefname{algorithm}{Alg.}{Alg.}
\crefname{figure}{Fig.}{Fig.}
\crefname{fact}{Fact}{Fact}
\Crefname{table}{Tab.}{Tab.}
\crefname{problem}{Problem}{Problem}
\crefname{appendix}{Appendix}{Appendix}

\usepackage{tikz} 

\hide{
\binoppenalty=700
\brokenpenalty=0 
\clubpenalty=0   
\displaywidowpenalty=0   
\exhyphenpenalty=50
\floatingpenalty=20000
\hyphenpenalty=50
\interlinepenalty=0
\linepenalty=10
\postdisplaypenalty=0
\predisplaypenalty=0 
\relpenalty=500
\widowpenalty=0  
}

\newcommand{\naive}{na\"ive}

\graphicspath{ {./figures/} }

\definecolor{forestgreen}{rgb}{0.13, 0.55, 0.13}

\newcommand{\basecase}{\theta}
\newcommand{\nheavy}{n_H} 
\newcommand{\nlight}{n_L} 
\newcommand{\nbucket}{n_B} 
\newcommand{\B}{l}

\newcommand{\stepname}[1]{{\textit{#1}}}
\newcommand{\block}{subarray\xspace}
\newcommand{\bucket}{bucket\xspace}

\newcommand{\record}{record}
\newcommand{\presum}{X}
\newcommand{\offsets}{\vname{offsets}}

\newcommand{\usedbits}{\vname{\gamma}}
\newcommand{\dtmerge}{\fname{DTMerge}}

\newcommand{\sampling}{{\stepname{Sampling}}\xspace}
\newcommand{\counting}{{\stepname{Distributing}}\xspace}
\newcommand{\refining}{\stepname{Recursing}\xspace}
\newcommand{\Merging}{\stepname{Dovetail Merging}\xspace}
\newcommand{\merging}{\stepname{dovetail merging}\xspace}

\newcommand{\bits}{\mathfunc{bits}}

\newcommand{\cbrt}[1]{\sqrt[3]{#1}}
\newcommand{\matrixtrans}{\intercal}

\newcommand{\impname}[1]{\textsf{#1}}
\newcommand{\oursortfull}{{\impname{DovetailSort}}}
\newcommand{\oursort}{{\impname{DTSort}}}
\newcommand{\plain}{{\impname{Plain}}}
\newcommand{\plss}{\impname{PLSS}\xspace}
\newcommand{\plis}{\impname{PLIS}\xspace}
\newcommand{\ipso}{\impname{IPS$^4$o}\xspace}
\newcommand{\regions}{\impname{RS}\xspace}
\newcommand{\raduls}{\impname{RD}\xspace}
\newcommand{\ipsra}{\impname{IPS$^2$Ra}\xspace}
\newcommand{\plmerge}{\fname{PLMerge}}

\newcommand{\uniform}[1]{\emph{Unif-#1}}
\newcommand{\exponential}[1]{\emph{Exp-#1}}
\newcommand{\zipfian}[1]{\emph{Zipf-#1}}
\newcommand{\bitexp}[1]{\emph{BExp-#1}}
\newcommand{\bitsexponential}[1]{Bit-Exponential}
\newcommand{\bexp}{BExp}

\begin{document}


\title{Parallel Integer Sort: Theory and Practice}

\settopmatter{authorsperrow=4}
\author{Xiaojun Dong}
\affiliation{%
  \institution{UC Riverside}
  \country{}
}
\email{xdong038@ucr.edu}

\author{Laxman Dhulipala}
\affiliation{%
  \institution{University of Maryland}
  \country{}
}
\email{laxman@umd.edu}

\author{Yan Gu}
\affiliation{%
  \institution{UC Riverside}
  \country{}
}
\email{ygu@cs.ucr.edu}

\author{Yihan Sun}
\affiliation{%
  \institution{UC Riverside}
  \country{}
}
\email{yihans@cs.ucr.edu}


\begin{abstract}
Integer sorting is a fundamental problem in computer science.
This paper studies parallel integer sort both in theory and in practice.
In theory, we show tighter bounds for a class of
existing practical integer sort algorithms, which provides a solid
theoretical foundation for their widespread usage
in practice and strong performance.
In practice, we design a new integer sorting algorithm,
\textsf{DovetailSort}, that is theoretically-efficient and has good practical performance.

In particular, \textsf{DovetailSort} overcomes a common challenge in existing
parallel integer sorting algorithms, which is the difficulty of detecting and
taking advantage of duplicate keys.
The key insight in \textsf{DovetailSort} is to combine algorithmic ideas from
both integer- and comparison-sorting algorithms.
In our experiments, \textsf{DovetailSort} achieves competitive or better
performance than existing state-of-the-art parallel integer and
comparison sorting algorithms on various synthetic and real-world
datasets.
%
%
%
%
%
%
\end{abstract}



\renewcommand\footnotetextcopyrightpermission[1]{} 
\fancyhead{} 

\maketitle

\pagenumbering{arabic}

\section{Introduction}\label{sec:intro}

Sorting is one of the most widely-used primitives in algorithm design, and has been extensively studied.
For many if not most cases, the keys to be sorted are fixed-length integers.
Sorting integer keys is referred to as the \emp{integer sort} problem.
An integer sorting algorithm takes as input $n$ records with integer keys in the range $0$ to $r-1$,
and outputs the records with keys in non-decreasing order.
Despite decades of effort in studying integer sorting
algorithms, however, obtaining \emph{parallel integer sorting (IS) algorithms that are
 efficient both in theory and in practice} has remained elusive.

\myparagraph{Theoretical Challenges.}
\emph{As a special type of sorting, integer sorting algorithms \textbf{can} outperform comparison sorting algorithms by using the integer encoding of keys.}
This claim is verified in many existing studies~\cite{obeya2019theoretically,axtmann2022engineering} as well as our experiments (see \cref{fig:heatmap}). 
As a result, in real-world applications
(e.g.,~\cite{gbbs2021,ptreedb,wang2023parallel}), integer
sort is usually preferred (instead of comparison sort) when the keys are integers. 
While it is not surprising that IS algorithms can outperform comparison sorts,
we observe a \emph{significant gap in connecting the high performance with theory}.
Theoretical parallel IS algorithms with good bounds~\cite{bhatt1991improved,albers1997improved,andersson1998sorting,hagerup1991constant,matias1991parallel} are quite complicated, and we are unaware of any implementations of them.
On the other hand, for the practical parallel IS solutions~\cite{blelloch2020parlaylib,obeya2019theoretically,axtmann2022engineering}, we are unaware of ``meaningful'' analysis to explain their good performance for general $r$.
The best-known bounds for them, as discussed in~\cite{obeya2019theoretically}, are $O(n\log r)$ work (number of operations) and $\polylog(rn)$ span (longest dependence chain).
However, note that the main use case for integer sort is when $r=\Omega(n)$,
since otherwise the simpler counting sort~\cite{vishkin2010thinking,RR89} can be used.
In this case, the bounds for practical parallel IS algorithms are no better than comparison sorts ($O(n\log n)$ work and polylogarithmic span~\cite{blelloch2020optimal,goodrich2023optimal,blelloch2010low,blelloch2020parlaylib}).
This leads to the following open question in theory: do the practical parallel IS algorithms indeed have lower asymptotic costs than comparison sort (and if so, under what circumstances)?

\myparagraph{Practical Challenges.}
\emph{As a special type of sorting, integer sorting algorithms \textbf{should} outperform comparison sorting algorithms by using the integer encoding of keys.}
Since integers are comparable, comparison sort can be considered as a baseline for sorting integers.
Unfortunately, SOTA parallel IS algorithms do not consistently outperform comparison-based sorting algorithms.
One key reason is the inherent difficulty of dealing with \emph{duplicate keys} in integer sorts.
In principle, duplicate keys are beneficial for sorting algorithms.
For example, {\em samplesort} can skip a recursive subproblem between two equal pivots; similarly, quicksort can separate keys equal to the pivot to avoid further processing them.
Interestingly, such a case does not apply to integer sort.
Existing parallel IS implementations follow the {\em most-significant digit (MSD)} framework that partitions all
keys into buckets based on the \emph{integer encoding} (i.e., 8--12 highest bits), and recurses within each bucket.
As such, equal keys cannot be detected until the last recursion.
Although some techniques can be used to detect special distributions (e.g., all keys are the same),
to the best of our knowledge, no existing parallel IS implementation can benefit from duplicate keys in a general (provable) and non-trivial manner\footnote{
Some techniques in existing IS implementations~\cite{obeya2019theoretically,axtmann2022engineering} have a side-effect to benefit from duplicates in some cases, at a cost of making the algorithms \emph{unstable}.  We want to overcome this issue without sacrificing stability.
}.
\cref{fig:heatmap,tab:synthetic} show that on inputs with heavy duplicates, SOTA parallel IS algorithms can be slower than comparison sorts.

\myparagraph{Our Contributions.} In this paper, we study parallel integer sort to overcome the aforementioned challenges both in theory and in practice.
In theory, \emph{we answer the open question by proving that \textbf{a class of} parallel IS algorithms, including \oursortfull{} proposed in this paper, have $O(n\sqrt{\log r})$ work for input size $n$ and key range $0$ to $r-1$}.
This is asymptotically better than comparison sorts for a large key range $r=n^{o(\log n)}$.
This explains why the existing parallel IS algorithms outperform comparison sorts in practice in many cases.
We also show that we can achieve $\polylog(nr)$ span \whp{} for an unstable integer sort, and $\tilde{O}(2^{\sqrt{\log r}})$ span for a class of practical stable parallel IS algorithms, such as our new \oursortfull{} and the integer sort in ParlayLib~\cite{blelloch2020parlaylib}. 
We also prove that \oursortfull{} has $\Theta(n)$ work for some special input distributions with heavy duplicates.
In addition to the new bounds, our theoretical results (\cref{sec:theory})
also explain the choices of parameters for many (if not all) practical parallel IS implementations.
We believe that our theoretical contribution is important both for
analyzing our own algorithm and for understanding parallel integer
sort as a general algorithmic problem.

\begin{figure}
  \centering
  \includegraphics[width=\columnwidth]{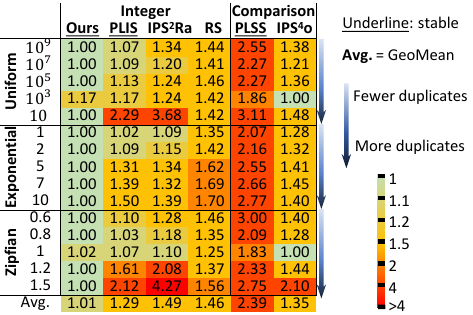}
  \caption{\textbf{Heatmap to compare sorting algorithms on $10^9$ records with 32-bit keys and 32-bit values.}
  All numbers are \textbf{running times relative to the best} for each input.
  Raw data are in \cref{tab:synthetic}.
  The baseline algorithms are described in \cref{tab:baseline}.
  }\label{fig:heatmap}
\end{figure} 

In practice, \emph{we propose \oursortfull{} (\oursort{}), a parallel integer sorting algorithm that is \textbf{efficient and stable}}. 
\oursort{} combines the algorithmic ideas from both integer- and comparison-sorting algorithms.
The goal is to detect and take advantage of heavily duplicate keys (\emp{heavy keys}) as in samplesort (a commonly-used comparison sort),
but preserve the low asymptotic work bound as in integer sort.
The key idea are to 1) use sampling to identify heavy keys, separate them from the other keys (called \emp{light keys}),
and skip them in the recursive calls,
and 2) exploit the bitwise encoding of the keys to split light keys
into subproblems and deal with them recursively.
Although the high-level idea of \emph{separating} heavy and light keys is natural, an interesting challenge is that
all keys need to be put in increasing order in the output,
and thus the heavy and light keys have to be finally \emph{interleaved} with each other.
Therefore, the heavy-light separation idea requires \emp{dovetailing} the keys at a low cost.
To do this, \oursort{} first carefully puts the heavy keys close to their final destinations when they are identified.
After the light keys have been sorted from the recursive calls, \oursort{} uses a dedicated algorithm to combine them into the final sorted order.
We present more details in \cref{sec:oursort}.
Due to the algorithmic improvements, \oursort{} achieves
consistently strong performance on various synthetic and real-world
input instances.
We evaluate \oursort{} using SOTA parallel sorting algorithms as baselines.
We illustrate the results on 32-bit keys in \cref{fig:heatmap}, and present more in \cref{sec:exp}.
Compared to the best baseline on 32-bit keys (\plis{}), our algorithm is competitive (1.02$\times$ faster) for no or light duplicate keys, and can be up to 2.3$\times$ faster with more duplicates.
The advantage of \oursort{} is larger on 64-bit keys (see \cref{tab:synthetic}).
We also evaluate \oursort{} in several applications with real-world input data.
In tens of the input instances we tested, \oursort{} is the fastest in most cases. 
We note that sorting is one of the most well-studied problems, 
and it is almost impossible for one algorithm to be the {best} for \emph{all} input instances.
The goal is to design algorithms that are always close to best in any use case.
We believe \oursort{} achieves very good overall performance on a variety of input patterns.
Our code is publicly available~\cite{integersortcode}.
\ifconference{
We present more results and analyses in the full version of this paper~\cite{dong2024parallelfull}.
}

\section{Preliminaries}\label{sec:prelim}


\newcommand{\mysubsection}[1]{#1.}
\titleformat{\subsection}[runin]
{\normalfont\normalsize\bfseries}{\thesubsection}{0.5em}{\mysubsection}

\subsection{Notations}
We use $O(f(n))$ \emph{with high probability} (\whp{}) in $n$ to mean $O(cf(n))$ with probability at least $1-n^{-c}$ for $c \geq 1$.
Throughout the paper, we consider the high probability bound in $n$ for input size $n$. We omit ``in $n$'' with clear context.
We use $[r]$ to refer to integers in the range $0$ to $r-1$.
In our analysis, we use the term $2^{\sqrt{\log r}}$, which satisfies
$\omega(\log^{c_1}r)=2^{\sqrt{\log r}}=o(r^{c_2})$ 
for any positive constants $c_1$ and $c_2$.
We use $\tilde{O}(f(nr))$ to denote ${O}(f(nr)\cdot \polylog(nr))$.
We use $\log n$ to represent $\log_2 n$.
More notations are presented in \cref{tab:notation}.

\begin{table}
  \centering \small
  \begin{tabular}{Bl}
    \hline
    $A[1..n]$ & Original input array with size $n$\\
    $[r]$ & The range of the integer keys $0, \dots, r-1$ \\
    $\gamma$ & Number of bits in a ``digit'' (sorted by each level)\\
    $b=2^\gamma$ & The ``radix'' size\\
    $\basecase$ & Base case size threshold\\
    $n'$ & Current (recursive) problem size\\
    $d$ & Number of remaining ``digits'' to be sorted\\
    \hline
  \end{tabular}
  \caption{\small \textbf{Notations and parameters in this paper.}
  }\label{tab:notation}
\end{table}

\subsection{Computational Models}
We use the work-span (or work-depth) model for fork-join parallelism with
binary forking to analyze parallel algorithms~\cite{CLRS,blelloch2020optimal}.
We assume a set of \thread{}s that share a common memory.
A process can \forkins{} two child software \thread{s} to work in parallel.
When both children complete, the parent process continues.
The \defn{work} of an algorithm is the total number of instructions and
the \defn{span} (depth) is the length of the longest sequence of dependent instructions in the computation.
We can execute the computation using a randomized work-stealing
scheduler in practice in $W/P+O(D)$ time \whp{} on $P$ processors~\cite{BL98,ABP01,gu2022analysis}.


\subsection{Sorting and Integer Sorting}\label{sec:integer}

Integer sort takes as input an array $A$ of $n$
records with a key function that maps each record to an integer key in range $[r]$.
The goal is to reorder the records to be in non-decreasing order.
As mentioned in \cref{sec:intro}, the main and most general use scenario of integer sorts is when $r=\Omega(n)$, so we assume this in our theoretical analysis
unless otherwise specified. 
Many existing integer sort algorithms consider the integer keys as
radix-$b$ (base-$b$) numbers, and work on each ``digit'' in a round.
Usually each digit refers to every $\usedbits$ bits in the integer and $b=2^{\usedbits}$.
These algorithms are also referred to as radix sort.
The classic sequential radix sort uses $b=\Theta(n)$, starts from the least significant digits (LSD),
and reorders all the records on every digit using a stable sort, until all bits are used up~\cite{CLRS}.
It takes $O(n\log_n r)$ work.

Unlike the sequential setting, existing parallel integer sort implementations start from the most significant digit (MSD).
The algorithms (see \cref{alg:msd}) have two steps: 1) \emph{distributing} all records into $b=2^{\usedbits}$ buckets by their MSD,
and 2) \emph{recursing} to sort all buckets in parallel.
This top-down approach allows all subproblems to be solved independently in parallel, which is more I/O friendly.
The distribution step can be performed by a counting sort (see \cref{sec:counting}) with the bucket id as the MSD.
Most practical parallel integer sorting algorithms~\cite{blelloch2020parlaylib,obeya2019theoretically,axtmann2022engineering,kokot2017sorting,cho2015paradis,Blelloch91,lee2002partitioned,polychroniou2014comprehensive,solomonik2010highly,wassenberg2011engineering,zhang2012novel} follow the framework in~\cref{alg:msd} and have good performance.
However, if a tight bound exists for these parallel MSD algorithms remains open.
\hide{While these algorithms have good performance, we are unaware of any non-trivial theoretical analysis to explain why they perform well.
The best known work bound for parallel MSD sort is $O(n\log r)$~\cite{obeya2019theoretically}, which shows no advantage over existing samplesort~\cite{blelloch2020optimal,goodrich2023optimal,blelloch2010low}.}

\hide{
There are parallel integer sorting algorithms with good theoretical bounds~\cite{bhatt1991improved,albers1997improved,andersson1998sorting,hagerup1991constant,matias1991parallel}, but unfortunately the underlying techniques are complicated, which limits their practicality and/or programmability.
}

For simplicity, we assume to store the sorted result in the input array (``in-place'').
Here ``in-place'' refers to the algorithm {\em interface}.
Our algorithm still uses $O(n)$ extra space.
A sorting algorithm is \defn{stable} if the output preserves the relative order among equal keys as in the input, and unstable otherwise.
Stability is required in many sorting applications.

\subsection{Counting Sort (aka. Distribution)}\label{sec:counting}

Counting sort is a commonly-used primitive in parallel sorting algorithms to distribute records into (a small number of) buckets.
It takes an input of $n'$ records with a function to map each record to an integer (the \emp{bucket id}) in $[r']$ ($r'\le n'$).
The goal is to sort the records by the bucket ids.
The classic stable counting sort~\cite{vishkin2010thinking}, which is widely used in practice, has $O(n')$ work and $O(r'+\log n')$ span.
Due to page limits, we put more discussions about this algorithm in \ifconference{the full paper~\cite{dong2024parallelfull}}\iffullversion{\cref{sec:app-counting}}, but the audience can treat the counting sort as a black box.

There exists a randomized algorithm for unstable counting sort on $n'=\Omega(\log^2n)$ records with $r'=O(n'\log^{O(1)}n)$ buckets in $O(n'+r')$ work and $O(\log n)$ span \whp{} in $n$~\cite{RR89}.
However, this approach is complicated and is less frequently used in practice (discussions also given in \ifconference{the full paper~\cite{dong2024parallelfull}}\iffullversion{\cref{sec:app-counting}}).

\begin{algorithm}[t]
\small
\caption{Parallel-MSD-Sort($A[0..n-1], d$)\label{alg:msd}}
\SetKwFor{parFor}{parallel\_for}{do}{endfor}
\KwIn{$A$: input array with keys in $[r]$. $d$: number of digits remained to be sorted. Each digit refers to $\usedbits$ bits in the key. On the root-level $d=(\log r)/\usedbits$. $\basecase$: base case threshold.}
\notes{We use the $d$-th digit as the $d$-th digit from low to high. }

\lIf(\tcp*[f]{All bits are sorted}) {$d=0$} {
    \Return{$A$} \DontPrintSemicolon \label{line:nobits}
}
\lIf {$|A|\le \basecase$} {
    \Return{ComparisonSort$(A)$\label{line:msd-base-case}}\tcp*[f]{Base case} \DontPrintSemicolon 
}


\emph{Distributing:} Apply counting sort on $A$ using the $d$-th digit of the key as bucket id\label{line:msd-distr}\\
\emph{Recursing:} For $i\in [2^{\usedbits}]$, let $A_i$ be the chunk in $A$ where the records have the $d$-th digit as $i$ \\
\lparFor{$i\gets 0$ to $2^{\usedbits}-1$\label{line:msd-par-for}}{
    Parallel-MSD-Sort($A_i$, $d-1$)
}

\end{algorithm}

\subsection{Comparison Sort, Semisort, and Sampling}\label{sec:sort}

Efficient and parallel sorting is widely studied~\cite{inoue2015simd,shen2004adaptive,bhattacharya2022cache,traff2018parallel,katajainen1997meticulous,sato2016chunking,shamoto2016gpu}.
Comparison sort and semisort are the other two types of sorting problems.
Comparison sort takes an input of $n$ records and a comparison function ``$<$'' to order two given records, and outputs the records in non-decreasing order.
Semisort takes an input of $n$ records and an equality-test ``$=$'' on two records, and reorders the records such that records with equal keys are adjacent in the output~\cite{dong2023high}.
SOTA parallel algorithms for comparison sort~\cite{axtmann2022engineering,blelloch2020parlaylib} and semisort~\cite{gu2015top,dong2023high} all use sampling that can detect heavy duplicate keys.

The sampling scheme used in these papers
was first proposed in~\cite{RR89}.
Given a parameter $p$ and input size $n$, the sampling scheme first selects $p\log n$ samples uniformly at random.
It then sorts the samples, subsampling every $(\log n)$-th element into a list $S'$.
Based on Chernoff bounds, any key that appears more than once in $S'$
must have $\Omega(n/p)$ occurrences in the input \whp{}.
In existing sorting/semisorting algorithms~\cite{axtmann2022engineering,blelloch2020parlaylib,gu2015top,dong2023high},
each such key will be placed in its own bucket.
Since all records in this bucket have the same key, no further sorting is needed.
As such, these algorithms can be significantly faster when the input contains many duplicate keys.
Our \oursort{} is inspired by this technique for detecting heavy duplicate keys.
In \cref{sec:oursort}, we will discuss how to apply this sampling scheme to an integer sort algorithm.

\titleformat{\subsection}
{\normalfont\large\bfseries}{\thesubsection}{1em}{}

\begin{figure*}
  \centering
  \includegraphics[width=\textwidth,page=1]{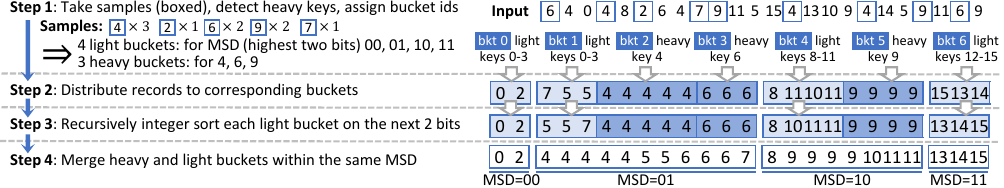}
  \caption{\textbf{An overview of the approach in the \oursort{}.} Here $r=16, \usedbits=2$.
  For simplicity and space limit, the sampling scheme in the figure is not exactly accurate as described in the algorithm.
  Here we simply set keys with 2 or more samples as heavy keys.
  }\label{fig:sort}
\end{figure*}

\section{The \oursortfull{} Algorithm}\label{sec:oursort}

We now present \oursortfull{} (\oursort{}), which follows the MSD sort framework in \cref{alg:msd}.
By a careful design, \oursort{} can efficiently handle heavy duplicate keys, and is \emp{stable}.
%

To detect and handle duplicate keys, we use the sampling idea in existing comparison-based sorts.
However, we need to carefully combine this approach with the MSD sort to avoid blowing up the work bound to $O(n\log n)$.
As mentioned, we want our algorithm to separate heavy and light keys, with the goal of avoiding sorting heavy keys in all recursive levels.
The idea of identifying heavy/light buckets has been used in parallel semisort~\cite{gu2015top,dong2023high}.
However, integrating it into a parallel MSD sort is highly non-trivial. 
Unlike in semisort where heavy and light keys are completely independent and can be placed in any order in the output,
all buckets in an integer sort must be finally placed in an increasing order.
We refer to this step as the ``\emph{\merging{}}'' step, and discuss an efficient solution both in theory and in practice in \cref{sec:merge}.

\begin{algorithm}[t]
\small
\caption{\oursort($A,d$)\label{alg:intsort}}
\KwIn{$A$: input array with keys in $[r]$. $d$: number of digits remained to be sorted.
Each digit refers to $\usedbits=\sqrt{\log r}$ bits in the key. On the root-level $d=(\log r)/\usedbits=\sqrt{\log r}$.}
\notes{We use ``the $d$-th digit'' as the $d$-th digit from low to high. }
\SetKwInOut{Parameters}{Parameters}
\SetKwProg{myfunc}{Function}{}{}
\SetKwFor{parForEach}{parallel\_for\_each}{do}{endfor}
\SetKwFor{parFor}{parallel\_for}{do}{endfor}

\DontPrintSemicolon


\lIf(\tcp*[f]{All bits are sorted}) {$d = 0$} {
    \Return{$A$} \DontPrintSemicolon \label{line:nobits}
}
\lIf(\tcp*[f]{Base cases}) {$|A| \le 2^{2\usedbits}$} {
    \Return{ComparisonSort$(A)$ }  \DontPrintSemicolon\label{line:basecase}
}

\stepcomment{ \sampling{} (find heavy keys \& assign buckets)}
$S \gets \Theta(2^\gamma\log n)$ sampled keys from $A$ \label{line:sample}\\
Sort $S$, subsample every ($\log n$)-th key, and store the keys that have more than one subsamples into $S'$\\
\For(\tcp*[f]{each heavy keys' MSD and its key}) {$i\gets 0$ to $|S'|-1$} {
  $h[i]\gets \langle$the $d$-th digit in $S'[i]$, $S'[i] \rangle$\label{line:h}
}
\For(\tcp*[f]{each light bucket's MSD and a dummy key}) {$i\gets 0$ to $2^{\usedbits}-1$} {
  $l[i]\gets \langle i,-1 \rangle$\label{line:l}
  }
Merge $h$ and $l$ into an array $\mathit{buckets}$\\
Initialize the hash table $H$ and lookup table $L$\\
\For(\tcp*[f]{assign bucket ids in order}) {$i\gets 0$ to $|\mathit{buckets}|$} {\DontPrintSemicolon
  $\langle x, y \rangle \gets \mathit{buckets}[i]$\\
  \lIf (\tcp*[f]{light bucket and its id}) {$y=-1$}{
    $L[x]\gets i$
  } \lElse(\tcp*[f]{heavy bucket and its id}){
    Insert key $y$ with value $i$ to $H$\label{line:sample-end}
  }
}

\hide{
\tcp{We only compare the $\usedbits$ bits in an integer after this step.}
Initialize the heavy table $H$ and light table $L$\\
$\nlight \gets 2^{\usedbits}$\\
$id \gets 0$\\
\For {$i: 0\le i< \nlight$} {
    $L[i] \gets id$\\
    $id \gets id+1$\\
    \For {each key $k \in S$ between $[2^i, 2^{i+1})$} {
        \If {the occurrences of $k$ in $S$ is at least $\log n'$} {
            $H.$insert$(k,id)$ \label{line:insert}\tcp*[f]{Assign bucket id $i$ to heavy key $k$}\\
            $id \gets id+1$\\
        }
    }
}
}

\hide{
$\nheavy \gets$ number of distinct keys in $H$ \label{line:step1:end}\\
$\nbucket \gets \nlight + \nheavy + 1$ \tcp{Total number of buckets, including ight, heavy, and the overflow buckets}
}

\stepcomment{\counting{} (Reorder $A$ by bucket ids)}

Use counting sort on $A$ with key function $\mathsf{GetBucketID}$ (\cref{line:bucketid})\label{line:distribute}\\

\stepcomment{\refining{} (sort light buckets)}
\lparForEach{light bucket $B$} {
    \oursort($B,d-1$)\label{line:recursive}\DontPrintSemicolon
}

\stepcomment{Dovetail Merging (interleave light and heavy buckets)}
\parFor(\tcp*[f]{merge buckets for all MSD zones}) {$i\in [2^{\usedbits}]$} {
    Let $B_0, B_1,\dots B_m$ be all buckets in MSD zone $i$\\
    \dtmerge($B_0, B_1,\dots B_m$)\label{line:dtmerge}
}

\Return{$A$}\\

\vspace{.05in}
\myfunc{\upshape \textsf{GetBucketId}$(k)$\label{line:bucketid}} {
    \leIf {$k$ is found in $H$} {
        \Return{$H[k]$} \DontPrintSemicolon
    }  {
        \Return{$L[k]$} \DontPrintSemicolon
    }
}

\end{algorithm}

We present \oursort{} in \cref{alg:intsort} and an illustration in \cref{fig:sort}.
It consists of four steps, explained in the subsections below.
It is an MSD sort using $\usedbits=\sqrt{\log r}$ bits in each digit.
We explain how this parameter enables a good work bound in \cref{sec:theory}.
We use $n$ as the original input size, and $n'$ as the current recursive subproblem size.
Separating them is needed to rigorously analyze the algorithm and give the correct high probability bounds.
We use \emph{heavy/light bucket}  to mean a bucket for heavy/light keys, and \emph{heavy/light record}  to mean a record with a heavy/light key.
We use ``MSD'' to refer to the MSD in the current subproblem (see the $d$-th digit in \cref{alg:intsort}).

\hide{
\begin{enumerate}[leftmargin=*]
    \item \textbf{\sampling{}.} The goal of this step is to detect the heavy keys in this level, and assign each light and heavy key with an associated bucket id. Unlike the semisort algorithm where heavy keys are put aside and fully separated with the light keys, in \oursort{} we try to keep the heavy keys close to their final destination from this very first step.
    \item \textbf{\counting{}.}
    \item \textbf{\refining{}.}
    \item \textbf{\merging{}.}
\end{enumerate}}

\hide{Our algorithm is an MSD algorithm, which means it sorts the bit range $[0, b-1]$ from
the most significant bit (i.e., $b-1$) to the least significant one.
Our algorithm is recursive and works on $\usedbits$ bits (or fewer in the last level) in each level.
To simply the algorithm description,
we present our algorithm in the top level (with a bit range $[b-\usedbits, b-1]$) with clearly mark recursive call.
We also assume the total input size is $n$ and the input size of the current level is $n'$.
}

\subsection{Step 1: \sampling{}}\label{sec:sample}
As mentioned earlier, the first step is to use sampling to determine heavy keys, and assign a bucket id to each heavy and light bucket.
The theory of this idea is developed in previous work on samplesort and semisort~\cite{RR89,gu2015top}.

In this step, $\Theta(2^{\gamma}\log n)$ samples are picked into an array $S$.
Recall that $\gamma$ is the number of bits to be sorted in a level. 
We then subsample every $(\log n)$-th key in $S$, and output keys with more than one subsample to $S'$.
As discussed in \cref{sec:sort}, we define all keys in $S'$ as heavy keys, and all other keys are light keys.
According to the MSD (the highest $\usedbits$ bits) of the keys, the key range is divided into $b=2^\usedbits$ subranges.
We call each key range sharing the same MSD an \emp{MSD zone}.
In a regular integer sort, each MSD zone forms a bucket.
In our case, all \emph{light keys} within the same MSD zone will share a bucket, but each \emph{heavy key} has its own bucket.

Unlike samplesort and semisort that do not need to further handle heavy keys once they are recognized and distributed, in integer sort,
we need to finally combine them with light keys in the same MSD zone.
Hence, we wish to keep the heavy keys close to their final destination in the output from this very first step,
in order to perform an efficient merge at the end of the algorithm.
To do so, we will set the bucket id of a heavy bucket immediately after the light bucket which it should finally be integrated into.
Namely, for each MSD zone, there is exactly one light bucket followed by some (possible zero) heavy buckets, which are all the heavy keys with the same MSD as this zone.
An illustration is given in \cref{fig:sort}.
A unique id is assigned to each bucket in a serial order.
Multiple heavy buckets in the same MSD zone will be totally ordered by their full keys. 
In this way, the light buckets and heavy buckets in the same MSD zone are consecutive,
which makes it easier to combine them in Step 4.

We explain this step in detail on \cref{line:sample}--\ref{line:sample-end} in \cref{alg:intsort}.
To assign the bucket ids, we first obtain an array $h$ of all heavy keys in a sorted order (\cref{line:h}),
which is a pair $(z,k)$, where $k$ is the key itself and $z$ is the MSD zone of $k$.
We then use an array $l$ (\cref{line:l}) to maintain a pair $(z,k)$ for each light bucket (similar to $h$),
where $z\in[2^{\usedbits}]$ is the MSD of the light bucket, and $k=-1$ is a dummy key to mark this is a light bucket.
By merging $h$ and $l$, all buckets will be re-organized in the desired order---all buckets will first be ordered by their MSD;
within the same MSD, the first bucket is always the light one, followed by all heavy buckets sorted by their key.
We can then assign ids to them serially. 

After assigning the bucket ids, we build a hash map $H$ to maintain the bucket id of the heavy keys, and use a lookup array $L$ to map
every MSD to its associated light bucket id.
As mentioned, the total number of buckets is small, so the cost to maintain them is small both in time and space.
With $H$ and $L$, we can determine the bucket id for any key $k$ in constant time (the \textsf{GetBucketID} function in \cref{alg:intsort}):
we first look up whether $k$ is in $H$;
if so, $k$ is a heavy key, and its bucket id can be found in $H$.
Otherwise, $k$ is a light key, and we can look up its bucket id in $L$ by using the MSD of $k$.


\subsection{Step 2: \counting{}}

After the first step, we know the bucket id of all keys.
In this step, we use the stable counting sort in \cref{sec:counting} with their bucket id as the key,
similar to many existing parallel sorting algorithms~\cite{dong2023high,blelloch2020parlaylib},
After this step, all records are stably reordered into the corresponding buckets.

\subsection{Step 3: \refining{}}\label{sec:recursion}
After Step 2, keys are sorted across MSD zones, but
keys in the same MSD zone are not necessarily sorted.
In this step, we first sort each light bucket recursively.

As with a regular integer sort, the recursive call deals with the subproblem using the same integer sort routine on the next ``digit''.
In \oursort{}, since all records in a heavy bucket have the same key, no further sorting is needed.
Hence, we only recursively sort the light buckets (\cref{line:recursive}).

Note that by using a recursive approach, we can potentially detect much more than $2^\gamma$ heavy keys if we consider all levels of recursion.
Namely,
the heaviest keys will be identified at the topmost level and bypass all recursive calls,
and some medium-level heavy keys will be identified in the middle levels but still bypass a few recursive calls.
This enables a balance in finding more heavy keys \emph{in total} to save work,
as well as to use a limited number of heavy keys \emph{in each subproblem} to avoid a large amount of extra space in the counting sort in Step 2 (see a discussion in~\cite{dong2023high}).


\begin{figure}
  \centering
  \includegraphics[width=\columnwidth]{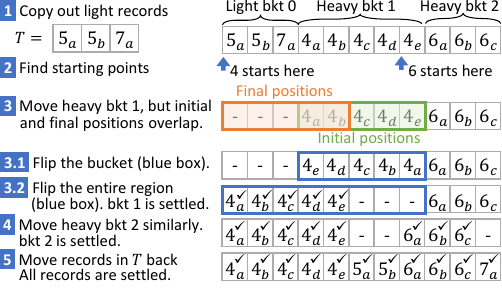}
  \caption{\textbf{Illustration of the dovetail merging step.} The example merges the buckets in MSD zone 01 in \cref{fig:sort}.
  We use a letter as subscription to distinguish different records with the same key.}\label{fig:merge}
\end{figure} 
\subsection{Step 4: \Merging{}}\label{sec:merge}
After Step 2, all keys are sorted across MSD zones,
and all heavy keys are ordered by their keys.
Step 3 further sorts all light keys within each MSD zone.
To compute the final output, we finally must merge the two sorted sequences (light records and heavy records) in each MSD zone.
We refer to this step as the \merging{} step, which is the \dtmerge{} function on \cref{line:dtmerge} in \cref{alg:intsort}.

Each \dtmerge{} call deals with an MSD zone with $m+1$ buckets $B_{0..m}$ stored consecutively in array $A$.
The first one $B_0$ is the light bucket in this zone.
The rest of the $m$ buckets are heavy buckets.
Given these special properties, the problem is equivalent to viewing the $m$ heavy buckets as partitioning the light bucket into $m+1$ blocks, and then to interleave them.
Hence, we call this process \merging{}.
A \naive{} solution is to call a standard parallel merge to combine $B_0$ with the rest $B_{1..m}$, with linear work and polylogarithmic span.
Despite the efficiency in theory, in practice, it requires two rounds of out-of-place global data movements---\emph{all} records
have to be merged into a temporary array, and then be copied back to the input array.
Such data movements will cause more cache misses and may make the algorithm inefficient.
There exists an in-place parallel merge algorithm in theory~\cite{gu2021parallel},
but it incurs a large constant hidden in the asymptotic cost (moving each record at most four times),
which may also be expensive in practice.
To enable good performance, our goal is to reduce the data movements,
with the additional consideration of moving less data out of the input array.

Our idea is based on two observations.
First, being strictly in-place is unnecessary, since $O(n)$ temporary space is required by the counting sort in Step 2.
Therefore, we can reuse the space while minimizing the number of out-of-place memory access.
The second useful observation is that the number of heavy keys is guaranteed to be small (at most $2^{\usedbits}$ of them).
As a result, to achieve a practical solution, our idea is to process
buckets in the same MSD zone one by one \emph{sequentially} (note that
all MSD zones are still processed in parallel),
and aim to achieve \emph{high parallelism within each bucket}.

With these insights, we present our \dtmerge{} algorithm with an illustration in \cref{fig:merge} and pseudocode in \cref{alg:dtmerge}.
The high-level idea is to only copy the smaller one of the light bucket ($B_0$)
or the heavy buckets ($B_{1..m}$) out to a temporary array, so that we copy at most half the records.
The larger half will be handled within the input array.
WLOG, assume the current subproblem has more heavy records than light ones (\cref{line:more-heavy}),
and the other case can be handled symmetrically (\cref{line:more-light}).
In this case, we first move the light records to a temporary array $T$ (\cref{line:merge-copy}),
and then move all heavy buckets to their final positions.
To obtain the final positions of each heavy bucket, we first binary search all heavy keys (in parallel) in the light bucket $B_0$ (\cref{line:binary-search}).
Combining this information with the sizes of all buckets, we can infer the
starting point of each heavy key. 
The next step is to (stably) move the entire bucket $B_i$ to this starting point,
and do so for all buckets one by one from $B_1$.

Ideally, we wish to copy all records in $B_i$ to their final positions in parallel,
but have to be careful to ensure that their final positions are ``safe'' to be overwritten.
Since all heavy buckets are initially at the end of this MSD zone,
their final destinations must be ahead of their initial positions.
There are three cases for a bucket $B_i$'s destination.
The first case is when it contains light records, which is safe since all light records have been backed up in $T$.
The second is when it overlaps with an earlier heavy bucket $B_{i'}$ where $i'<i$.
It is also safe---since $i'<i$, bucket $i'$ must have been moved to its final destination.
In the first two cases, we can safely move the records directly into the final destination in parallel (\cref{line:non-overlapping}).
The last and the most interesting case is when the destination of the bucket overlaps with its own initial position (\cref{line:overlap}).
Our idea is inspired by the in-place \emph{circular shift} algorithm~\cite{gries1981swapping,yang2013place} (see Phase 3 in \cref{fig:merge}):
we can first flip (reverse the order) the bucket itself (\cref{line:flip1}), and then flip the entire associated region in the array (\cref{line:flip2}).
Flipping an array can be performed in parallel and in-place by swapping the corresponding records.
As such, the entire chunk of $B_i$ is moved to an earlier position in a \emph{stable} way.
After dealing with all heavy buckets, we copy all records in $T$ back to $A$, all in parallel, directly to their final destination (\cref{line:merge-copy-back}).
In this way, our idea incurs out-of-place copies for at most half the records, and
requires moving the remaining records at most twice, but entirely within $A$.
Within each bucket, all operations (direct copy or swapping records) are highly parallel.

Note that although buckets within the same MSD zone are processed sequentially,
this only increases the span of the entire algorithm by a factor of $\log n$ because there are only $O(2^{\usedbits})$ heavy buckets.
Interestingly, although a parallel merge gives better span bound than our \dtmerge{}, even in theory it is not necessary,
because the counting sort in Step 2 also requires $O(2^{\usedbits})$ span.
We present more theoretical analysis in \cref{thm:msdours}.

In \cref{sec:exp-study}, we experimentally study \dtmerge{} and show it accelerates Step 4 compared to using a parallel merge.
We believe the ``dovetail'' step is a unique challenge arising from detecting heavy keys in integer sorts,
and our new technique solves it in an effective yet simple way.

\hide{
Since keys are already sorted across different ranges, we can merge each range in parallel.
WLOG, we assume the number light keys are more than the number of heavy keys.
We also assume there are $k$ heavy keys in the range,
then the light keys are split into $k+1$ subarrays.
We can use binary search to find out the boundaries $p[]$ of these subarrays,
and the boundaries $q[]$ of them after the heavy keys being merged.
For example, $[p[0], p[1])$ (resp. $[p[k], p[k+1])$) is the range of the first (resp. $(k+1)$-th) subarray before merging,
and $[q[0], q[1])$ (resp. $[q[k], q[k+1])$) is the range of the first (resp. $(k+1)$-th) subarray after merging.
We first copy all heavy keys into a temporary array,
and move the light keys starting from the $(k+1)$-th subarray to the first subarray.
To move the $i$-th subarray, if $[p[i], p[i+1])$ and $[q[i], q[i+1])$ do not overlap, we can easily move the keys.
Otherwise, we perform two flipping as described in \cref{fig:merge}. \xiaojun{The figure is not accurate. I'll update it.}
After moving all light keys, we can copy the heavy keys from the temporary array back to the original array.
}

\subsection{Base Cases}\label{sec:basecase}
For the recursive structure, we use two base case conditions.
The first is as in regular integer sort when all digits have been sorted (\cref{alg:intsort} \cref{line:nobits}).
The second is when the problem size is smaller than $2^{2\usedbits}$, where we apply a comparison sort.
Using comparison sort as a base case is standard in implementations for parallel integer sort.
Interestingly, instead of viewing it as an empirical optimization, our theoretical analysis indicates that this base case,
as well as the threshold to use this base case, are very important to guarantee theoretical efficiency. We present more details in \cref{sec:theory}.

\begin{algorithm}[t]
\small
\caption{\dtmerge($A,B_0, B_1,\dots B_m$)\label{alg:dtmerge}}
\KwIn{$m+1$ buckets consecutive in array $A$. The first one $B_0$ is a sorted light bucket.
The other $m$ buckets are heavy buckets sorted by their keys.}
\DontPrintSemicolon
Binary search all heavy keys in $B_0$, accordingly get $p_{1..m}$,
where $p_m$ is the starting index for heavy bucket $B_i$ when the array is fully sorted.\label{line:binary-search} \\
\If(\tcp*[f]{more heavy keys than light}){$|B_0|\le \sum_{i=1}^{m}|B_i|$ \label{line:more-heavy}} {
  copy $B_0$ to array $T$\label{line:merge-copy}\\
  \For (\tcp*[f]{starting from the first heavy bucket}){$i\gets 1$ to $m$}{
    \tcp{Move $B_i$ to the final positions starting from $A[p_i]$}
    \If {The final positions for $B_i$ overlap with the current positions of $B_i$\label{line:overlap}}{
    Let $s$ be the current start point of $B_i$\\
    flip($A,s,s+|B_i|$)\tcp*[f]{Flip the entire region}\label{line:flip1}\\
    flip($A,p_i,s+|B_i|$)\tcp*[f]{Flip the destination region}\label{line:flip2}\\
    }
    \Else(\tcp*[f]{No overlap. Directly copy}){
    \parForEach{$j\gets 0$ to $|B_i|-1$\label{line:non-overlapping}}{
    Copy the $j$-th record in $B_i$ to $A[p_i+j]$
    }
    }
  }
  Move light records from $T$ to their final positions in $A$\label{line:merge-copy-back}
}
\lElse{
Symmetric case: copy heavy keys out and merge\label{line:more-light}
}

\end{algorithm}

\section{The Theory of Parallel MSD Sort}\label{sec:theory}

In this section, we provide an analysis for \emph{a class of MSD integer sorting algorithms} based on the framework in \cref{alg:msd}, including \oursort{}.
Recall that the algorithm sorts a ``digit'' of $\usedbits$ bits in each round by 1) distributing all records to the corresponding buckets by the current digit being sorted,
and 2) sorting each bucket recursively.
Finally, a base case of using comparison sort is applied to subproblems with sizes smaller than a parameter $\basecase$.
As discussed in \cref{sec:intro,sec:prelim}, although this parallel MSD sort is widely used in practice, we are unaware of any tight analysis of it.
To the best of our knowledge, the best work bound is $O(n\log r)$~\cite{obeya2019theoretically}, which is no better than a comparison sort with $O(n\log n)$ work for the realistic setting when $r=\Omega(n)$.
This is inconsistent with the excellent performance of existing parallel IS algorithms in practice, which are much faster than SOTA comparison-based sorts for many inputs.
In this paper, we bridge this gap between theory and practice by showing \cref{thm:msd,thm:msd2,thm:msdours}.
Recall that the main use case for integer sorts is when $r=\Omega(n)$ (otherwise a counting sort has a better cost bound).
We hence make this assumption on $r$ in our analysis.

\subsection{The Analysis of the General MSD Sort}

As a warmup, and to illustrate some of the key theoretical ideas, we start by showing that we can achieve an \emph{unstable} MSD sort with low work and span.

\hide{
As mentioned in \cref{sec:prelim}, many exiting parallel algorithms for integer sort~\cite{blelloch2020parlaylib,obeya2019theoretically,axtmann2022engineering,kokot2017sorting} recursively sort the records based on the most-significant digits (MSD) of the keys.
Here we give a high-level outline of this approach in \cref{alg:msd} and refer to it as a parallel MSD sort.
It basically sorts $\usedbits$ bits in each round, where $\usedbits$ can either be a constant or a function of $|A|$ and $\bits$\footnote{We refer to this algorithm a MSD radix sort if we use a constant $\usedbits$ for all recursive levels, but there seem to be no restrictions on why we need to use a fixed $\gamma$ in all levels.}.
The main difference in these algorithms is how to distribute the records to the buckets (\cref{line:msd-distr}).
}

\begin{theorem}\label{thm:msd}
  There exists an unstable parallel MSD sorting algorithm with $O(n\sqrt{\log r})$ work and $O(\log r + \sqrt{\log r}\log n)$ span \whp{}.
\end{theorem}
Under the assumption of the standard RAM model that can handle integers with $O(\log n)$ bits (i.e., $r=n^{O(1)}$), the theorem implies $O(n\sqrt{\log n})$ work and $O(\log^{1.5} n)$ span \whp{}.

The key in our analysis is to properly set up the parameters in \cref{alg:msd}.
Once we do so, the analysis becomes relatively clean.
The parameters to use are digit width $\usedbits=\sqrt{\log r}$ and base case threshold $\basecase = 2^{c\usedbits}$ for some constant $c\ge 1$.
This indicates {$(\log r)/\usedbits=O(\sqrt{\log r})=O(\usedbits)$ recursive levels} in the entire algorithm.
We use the \emph{unstable counting sort}~\cite{RR89} (see \cref{sec:counting}) for the distribution step,
and use comparison sort with $O(n'\log n')$ work and $O(\log n')$ span~\cite{goodrich2023optimal} to handle base cases with size $n'$.
The two main parts in \cref{alg:msd} are the base cases (\cref{line:msd-base-case}) and the distribution step (\cref{line:msd-distr}).
We first analyze their work with the parameters shown above.

\begin{lemma}\label{lem:msd-base-case}
  With the chosen parameters, the total work for base cases (\cref{line:msd-base-case}) in \cref{alg:msd} is $O(n\sqrt{\log r})$.
\end{lemma}
\begin{proof}
  Let $n_1,\cdots, n_{k}$ be the base case sizes, where $\sum n_i=n$.
  Due to the base case condition, $n_i\le 2^{c\usedbits}$, so the work to sort it is $O(n_i\log n_i)=O(n_ic\usedbits)$.
  The total work for all base cases is therefore $O(nc\usedbits)=O(n\sqrt{\log r})$ ($c$ is a constant).
\end{proof}
\begin{lemma}\label{lem:msd-distr}
  With the chosen parameters, the distribution~step (\cref{line:msd-distr}) \cref{alg:msd} in all recursive calls has work $O(n\sqrt{\log r})$ \whp{}.
\end{lemma}
\begin{proof}
  We start with the total work in each recursive level (i.e., those with the same parameter $d$).
  Assume all subproblems in this level have sizes $n_1,\cdots, n_{k}$, where $\sum n_i\le n$.
  Using the unstable counting sort~\cite{RR89}, the distribution work is $\sum O(n_i+2^{\usedbits})$ \whp{} for each subproblem.
  Due to the base case condition, $n_i> 2^{c\usedbits}\ge 2^{\usedbits}$ since $c\ge 1$.
  Hence the total work in each recursive level is $\sum O(n_i)=O(n)$.
  Consider all $O(\sqrt{\log r})$ recursive levels, the total work is $O(n\sqrt{\log r})$.
\end{proof}

With the two lemmas, we can now prove \cref{thm:msd}.

\smallskip
\begin{proof}[Proof of \cref{thm:msd}]
  Combining \cref{lem:msd-base-case,lem:msd-distr} gives the work bound in \cref{thm:msd}.
  We now prove the span.
  The span for the base case is $O(\log 2^{c\usedbits})=O(\usedbits)$~\cite{goodrich2023optimal}.
  The total span for one recursive call is the sum of $O(\log n)$ \whp{} (for the distribution step)~\cite{blelloch2020optimal} and $O(\gamma)$ (\cref{line:msd-par-for}).
  Since there are $O(\sqrt{\log r})$ recursive levels, the total span bound is as stated.
\end{proof}

\hide{First, it using counting sort~\cite{RR89} or the semisort~\cite{gu2015top} is not stable, so the overall sorting result is also not stable.
Second, both of the sorting algorithms are not I/O-efficient, so using them will cause significant performance overhead in practice.
We now show how to overcome these issues.}

We note that \cref{thm:msd} holds for $\usedbits=\Theta(\sqrt{\log r})$, but for simplicity we omit the $\Theta(\cdot)$ notation here.

\myparagraph{Practical and Stable MSD Sorting Algorithms.}
\cref{thm:msd} shows the theoretical efficiency of parallel MSD algorithms.
However, the counting sort~\cite{RR89} used in the analysis is not stable, and is unlikely to be practical due to too many rounds of data movements~\cite{gu2015top}.
Many practical algorithms choose to use the stable counting sort~\cite{vishkin2010thinking}.
We analyze it as follows.

\begin{theorem}\label{thm:msd2}
  There exists a stable parallel MSD sorting algorithm with $O(n\sqrt{\log r})$ work and $O(2^{\sqrt{\log r}}\sqrt{\log r})$ span.
\end{theorem}

\begin{proof}
First, note that the work analysis remains the same as \cref{thm:msd}, since the stable counting sort has the same work as the unstable one.
The only difference is that the stable version has span related to the number of buckets, which is $O(\log n + 2^{\usedbits})$.
Assuming $r=\Omega(n)$, this is $O(2^{\usedbits})$.
The total span of the entire MSD sorting is therefore $O(2^{\sqrt{\log r}}\sqrt{\log r})$.
\end{proof}

This indicates the same work bound as the unstable version.
Although the span is slightly higher than that in \cref{thm:msd},
this version enables stability, and is more practical.
In fact, many existing practical parallel IS algorithms, such as \plis{}~\cite{blelloch2020parlaylib},
use a similar approach as in \cref{thm:msd2}.

\hide{
To achieve stability and I/O efficiency, we can consider the distribution algorithm from~\cite{blelloch2010low}, which distributes the records into the buckets in a \emph{stable} manner.
It is used in PLIS~\cite{blelloch2020parlaylib} and also in \oursort.
Here we overview the high-level idea and give more details given in \cref{sec:algo}.
Given the current input size $n'$, it partitions the input array into $n'/l$ subarrays each with size $l$.
For each subarray, we sequentially count the appearances of the keys into the $2^\gamma$ buckets, and in total we have $n'/l \times 2^\gamma$ such counts in a matrix (the counting matrix).
We then apply a prefix-sum on the counting matrix to get the offset of each bucket in each subarray, and finally distribute the records to their destinations in the buckets.
This algorithm uses $O(n'+n'2^\gamma/l)$ work and $O(\log n+l)$ span.
Plugging in $l=O(2^\gamma)=O(2^{\sqrt{\log r}})$ gives $O(n')$ work and $O(2^{\sqrt{\log r}})$ span.}

\myparagraph{Theory-Guided Practice.}
\cref{thm:msd2} shows $O(n\sqrt{\log r})$ work for practical parallel MSD sort.
This explains why the integer sort algorithms can outperform comparison sorts with $O(n\log n)$ work for realistic range of $r=n^{O(1)}$.
%
Our parameter $\gamma=\sqrt{\log r}$ used in the analysis is also the choice for most practical MSD algorithms.
When $r=2^{64}$, we have $\gamma=\sqrt{\log r}=8$ and $\basecase=2^{c\usedbits}=2^{8c}$, which roughly matches the existing algorithms
that empirically sort 8--12 bits in each recursive level,
and use base case sizes from $2^8$ to $2^{16}$~\cite{axtmann2022engineering,blelloch2020parlaylib,obeya2019theoretically}.

\hide{
\medskip
As a summary, we bridge the theory and practice of integer sort by giving this analysis of the parallel MSD sort.
It is the first rigid theoretical analysis that can explain the efficiency of the existing and our new integer sort algorithms, the asymptotic trends, and the choice of parameter in practice.
}

\subsection{The Analysis of \oursort{}}

In this section, we show that our \oursort{} with the additional steps of sampling and dovetail merging is theoretically efficient.
We first prove the following theorem.

\begin{theorem}\label{thm:msdours}
  The \oursort{} algorithm (\cref{alg:intsort}) is a stable integer sort with $O(n\sqrt{\log r})$ work and $\tilde{O}(2^{\sqrt{\log r}})$ span.
\end{theorem}
\begin{proof}
  \oursort{} can be viewed as an extension of the generic stable integer sort which we proved work/depth bounds for in \cref{thm:msd2}.
  Compared to the algorithm analyzed in \cref{thm:msd2}, \oursort{} has two additional steps: the sampling step and the dovetail merging step.
  We first show that the asymptotic work bound remains the same. Note that without sampling and merging, a subproblem with size $n'$ has $O(n')$ work.
  In \oursort{}, the total number of samples is $|S|=2^{\usedbits}\log n$.
  We use base case threshold $\basecase=2^{2\usedbits}$, and assume $r=\Omega(n)$,
  so the cost to sort the samples is $O(|S|\log |S|)=O(2^{\usedbits}\log n\cdot \usedbits)=o(2^{2\usedbits})=o(n')$.
  Other parts in the sampling step have work proportional to $|S|$, so the overall work in Step 1 is $o(n')$.

  In the dovetail merging step, we need $O(2^{\usedbits})$ binary searches on $n'$ records, so the total cost is $O(2^{\usedbits}\log n')$,
  which is $O(n')$ considering $n'\ge 2^{2\usedbits}$ and $r=\Omega(n)$.
  When moving a bucket of size $t$, the total work is $O(t)$ (either a direct copy or two rounds of flipping on at most $2t$ records).
  We also copy out and back at most $n'/2$ elements. Combining both parts, the work for the merging step is also $O(n')$. Therefore, compared to the algorithm we used in \cref{thm:msd2}, \oursort{} uses asymptotically the same amount of work, $O(n\sqrt{\log r})$.

  For the span, note that the algorithm used in \cref{thm:msd2} has $O(2^{\usedbits})$ span in each recursive call.
  Again consider the two additional steps. Taking samples in parallel and sorting them requires $O(\usedbits)$ span.
  We process $L$ and $H$ sequentially since there are at most $2^\usedbits$ heavy buckets and exactly $2^\usedbits$ light buckets.
  This gives $O(2^{\usedbits})$ span in each recursive call.

  For the dovetail merging step, we need to move at most $2^{\usedbits}$ heavy buckets.
  Moving each of them is either a direct copy, or two rounds of simple flipping. Both can be done by a parallel-for on all affected records, giving $O(\log n)$ span assuming binary forking. This gives $O(2^{\usedbits}\log n)=\tilde{O}(2^{\usedbits})$ span in one recursive call.
  Therefore the two additional steps both have $\tilde{O}(2^{\usedbits})$ span in each recursive call.
  Considering all $O(\sqrt{\log r})$ levels gives the stated span bound.
\end{proof}

We also note that for the dovetail merging step, if we use the simple parallel merge (i.e., our baseline solution), we can get the same span bound of $O(2^{\sqrt{\log r}}\sqrt{\log r})$ as in \cref{thm:msd2}.
Also, if we assume arbitrary forking (e.g., PRAM), we can also obtain the same span bound as in \cref{thm:msd2}.
To avoid the minor subtlety here, we hide the logarithmic terms in the theorem by the $\tilde{O}(\cdot)$ notation.
Our main claim is that \oursort{} can achieve $O(n\sqrt{\log r})$ work and is highly parallel.

\smallskip
In addition to the worst-case bound, we can also show that \oursort{} can achieve asymptotically better work bound on some real-world distributions with heavy duplicates.
We show two special cases that \oursort{} achieves linear work in \cref{thm:special1,thm:special2}.
Here the exponential distribution means the frequencies of keys following the probability density function $f(x;\lambda)=\lambda e^{-\lambda x}$ for $x>0$ and a parameter $\lambda>0$.

\begin{theorem}\label{thm:special1}
  \oursort{} has $O(n)$ work \whp{} if the input key frequency exhibits an exponential distribution with $\lambda e^{\lambda}\ge \bar{c}/2^{\gamma}$ for some $\bar{c}>1$. Here $\lambda>0$ is the parameter of the exponential distribution, which gives probability density function $f(x;\lambda)=\lambda e^{-\lambda x}$ for $x>0$.
\end{theorem}

\begin{proof}
  We first consider the exponential distribution, where the input follows the probability density function $f(x;\lambda)=\lambda e^{-\lambda x}$ for $x>0$ and a parameter $\lambda>0$.
  Since we are considering integer sort, the randomly drawn real number is rounded to an integer, and randomly mapped to an integer in $[r]$.
  \oursort{} can detect a large fraction of the records as heavy keys when the input is drawn from an exponential distribution.
  Particularly, the sampling scheme in \cref{alg:intsort} can detect keys with more than $\bar{c}n/2^{\gamma}$ occurrences as heavy keys \whp{} for some constant $\bar{c}>1$~\cite{gu2015top}, or equivalently, keys with frequencies higher than $\bar{c}/2^{\gamma}$.
  This indicates that an integer $x$ with $\lambda e^{-\lambda x}\ge \bar{c}/2^{\gamma}$ will be identified as heavy.
  The cumulative distribution function of the exponential distribution is $F(x;\lambda)=1-e^{-\lambda x}$ for $x\ge 0$.
  Thus, the total number of light records is no more than $n\cdot e^{-\lambda x}=\bar{c}n\gamma/2^{\gamma}=\bar{c}n\sqrt{\log r}/2^{\sqrt{\log r}}$ when plugging in $x$ from the above threshold for heavy keys.
  However, since we are considering integer sort (keys are rounded to integers), we need to make sure that $x\ge 1$, so we need the assumption that $\lambda e^{\lambda}\ge \bar{c}/2^{\gamma}$.
  Even though we assume the worst case that the rest of the records will be processed for all the rest $\sqrt{\log r}-1$ levels, the cost to process them is $O(\bar{c}n\sqrt{\log r}/2^{\sqrt{\log r}}\cdot \sqrt{\log r})=O(n\log r/2^{\sqrt{\log r}})=o(n)$.
  Hence, the total work for \oursort{} on exponential distribution under the assumption is $O(n)$.
\end{proof}

\begin{theorem}\label{thm:special2}
  \oursort{} has $O(n)$ work \whp{} if there are no more than $c'2^{\gamma}$ distinct keys, for some constant $c'<1$.
\end{theorem}

\begin{proof}
  For simplicity, we use the same notation $\bar{c}$ as the previous proof.
  Similarly, keys with frequencies higher than $\bar{c}/2^{\gamma}$ will be detected as heavy.
  The maximum number of light records is no more than $c'2^\gamma\cdot \bar{c}n'/2^{\gamma}<c'\bar{c}n'$ for a subproblem of size $n'$.
  By setting $c'<1/\bar{c}$, the number of light records will shrink by a constant fraction.
  The total cost in a recursive level is linearly proportional to the total size of all subproblems.
  Since the number of light records decreases geometrically across levels, the total work is asymptotically bounded by the root level, which is $O(n)$.
\end{proof}

\section{Implementation Optimizations}\label{sec:app-opt}

We now discuss some optimizations we find useful when implementing \oursort.

\myparagraph{Overflow Bucket.}
Note that integer sort takes the input keys from a range from $0$ up to $r$, but the actual input keys may be in a smaller range.
For example, although the key can be 32-bit integers, the actual range can be much smaller than $2^{32}$.
To take advantage of this, one possibility is to compute the maximal key by a parallel reduce, as in many existing implementations~\cite{kokot2017sorting,blelloch2020parlaylib}.
Here we discuss a simple yet effective optimization to avoid this computation.

Similar to the heavy/light separation of the keys, here we also use sampling to detect the key range.
In each subproblem, \oursort{} sets the upper key range $r'$ to be the largest value in the sample set $S$.
As such, \oursort{} skips leading zeros in $r'$ since creating buckets for these bits are wasteful.
Since the upper bound is computed by a sample set, \oursort{} creates an \emph{overflow bucket} to hold records
whose keys have at least a ``1'' for the skipped bits.
The size of the overflow bucket is usually very small.
We put the overflow bucket at the end, and comparison sort records in this bucket.
As such, we do not need to compute the maximal key in the input.

\hide{
Although we will save the input using $r$-bit integers,
they sometimes only use $k < r$ \xiaojun{Check if we use $k$ somewhere else.} bits,
and the top $k-r$ bits are all zeros.\laxman{$r-k$?}
This case is common when we need to sort integers in range $[0, n)$ where $n \leq 2^k$ and $k < r$.
To avoid sorting the leading zero bits, existing implementations~\cite{kokot2017sorting,blelloch2020parlaylib} often
(1) allow the user to pass the number of bits to be sorted, or
(2) compute the maximum value in the input, then integer sort from the most significant digit of the maximum value.
We adapt the first optimization to our interface and further generalize the second heuristic.

Instead of computing the maximum integer from the entire input, we use the maximum integer from the sample set,
and setup the MSD zones from the most significant unsorted bit (say, $k$) of it.
For example, if we need to sort two bits (i.e., 00, 01, 10, 11), and the maximum integer from the sample set is 1,
we can optimistically assume the most significant bit are all zeros,
and allocate two MSD zones (00, 01) instead.
However, this assumption could be wrong.
To deal with these ``overflow'' integers, we add an extra overflow bucket after the two MSD zones.
If an integer does not belong to the two MSD zones, we add it to the overflow bucket.
Since none of the samples belong to the overflow bucket, we can assume that its size is small.
We can simply use a comparison-based sort to sort the integers in this bucket.}

\myparagraph{Minimizing Data Movement.}
In \cref{alg:intsort}, for simplicity, in our description we assume we use the stable distribution mentioned in \cref{sec:counting} and write back the bucket to $A$ (see \cref{line:distribute}).
The second copy-back is not necessary and should be avoided, since the performance of sorting algorithms can suffer due to excessive data movement.
In our implementation, we actually maintain two arrays $A$ and $T$.
After the distribution, we use array $T$ for future execution, until we see another distribution or a dovetail merging, and we then distribute or merge the results back to $A$ (or the other array in the recursion).
This approach will reduce the total data movement.
Finally, if the data remains in $T$ at the end of the algorithm, we need to move them back to $A$ before finishing.

\section{Experiments}\label{sec:exp}

\begin{table}
    \centering\small
    \begin{tabular}{@{}l@{  }c@{  }@{  }c@{}c@{  }@{  }l@{}}
      \toprule
      \bf Name & \bf Stable & \bf In-place &\bf Type & \bf Notes\\
      \midrule
      \textbf{\oursort} & Yes & No & Integer &Our integer sort algorithm\\
      \textbf{PLIS} & Yes & No & Integer &ParlayLib integer sort~\cite{blelloch2020parlaylib}\\
      \textbf{IPS$^2$Ra} & No & Yes & Integer & IPS$^2$Ra integer sort~\cite{axtmann2022engineering}\\
      \textbf{RS} & No & Yes &  Integer & RegionsSort~\cite{obeya2019theoretically}\\
      \textbf{RD} & No & No & Integer & RADULS~\cite{kokot2018even}\\
      \textbf{PLSS} &  Y/N & Y/N & Comparison &ParlayLib sample sort~\cite{blelloch2020parlaylib}\\
      \textbf{IPS$^4$o} & No & Yes & Comparison & IPS$^4$o sample sort~\cite{axtmann2022engineering}\\
      \bottomrule
    \end{tabular}
    \caption{\small \textbf{Algorithms tested in our experiments.}
    There are two versions of \plss{}. Here we use the unstable but faster one.
    In-place means fully in-place ($o(n)$ extra memory).
    \label{tab:baseline}}
  \end{table} 
\begin{table*}[htbp]
\small
\centering
\vspace{-.1em}
\begin{tabular}{cc|rr@{}rr|r@{  }r||rr@{}rrr|r@{  }r}

      \multicolumn{2}{c|}{\multirow{3}[1]{*}{\textbf{Instances}}} & \multicolumn{6}{c||}{\textbf{32-bit key and 32-bit value pairs}} & \multicolumn{7}{c}{\textbf{64-bit key and 64-bit value pairs}} \\
      \multicolumn{2}{c|}{} & \multicolumn{4}{c|}{\textbf{Integer}} & \multicolumn{2}{@{  }c@{  }||}{\textbf{Comparison}} & \multicolumn{5}{c|}{\textbf{Integer}} & \multicolumn{2}{@{  }c@{  }}{\textbf{Comparison}} \\
      \multicolumn{2}{c|}{} & \multicolumn{1}{c}{\textbf{Ours}} & \multicolumn{1}{c@{  }}{\textbf{PLIS}} & \multicolumn{1}{@{  }c@{}}{\boldmath{}\textbf{IPS$^2$Ra}\unboldmath{}} & \multicolumn{1}{c|}{\textbf{RS}} & \multicolumn{1}{c@{}}{\textbf{PLSS}} & \multicolumn{1}{c@{  }||}{\boldmath{}\textbf{IPS$^4$o}\unboldmath{}} & \multicolumn{1}{c}{\textbf{Ours}} & \multicolumn{1}{c@{  }}{\textbf{PLIS}} & \multicolumn{1}{@{  }c@{}}{\boldmath{}\textbf{IPS$^2$Ra}\unboldmath{}} & \multicolumn{1}{c}{\textbf{RS}} & \multicolumn{1}{c|}{\textbf{RD}} & \multicolumn{1}{c@{}}{\textbf{PLSS}} & \multicolumn{1}{c@{ }}{\boldmath{}\textbf{IPS$^4$o}\unboldmath{}} \\
      \midrule
      \multicolumn{15}{l}{\textbf{Standard distributions:}} \\
      \midrule
      \multirow{5}[2]{*}{\begin{sideways}\textbf{Uniform}\end{sideways}} & \boldmath{}\textbf{10$\boldsymbol{^9}$}\unboldmath{} & \underline{.500} & .537  & .671  & .718  & 1.27  & .690  & \underline{.994} & 1.14  & 1.09  & 1.43  & 1.86  & 1.65  & 1.11 \\
            & \boldmath{}\textbf{10$\boldsymbol{^7}$}\unboldmath{} & \underline{.501} & .549  & .600  & .705  & 1.14  & .604  & \underline{1.03} & 1.15  & 1.06  & 1.71  & 1.86  & 1.47  & 1.05 \\
            & \boldmath{}\textbf{10$\boldsymbol{^5}$}\unboldmath{} & \underline{.478} & .542  & .595  & .696  & 1.08  & .653  & \underline{.859} & 1.22  & 1.01  & 1.36  & 2.66  & 1.35  & 1.10 \\
            & \boldmath{}\textbf{10$\boldsymbol{^3}$}\unboldmath{} & .506  & .505  & .538  & .613  & .805  & \underline{.432} & .795  & 1.41  & 1.66  & 1.54  & 3.23  & 1.01  & \underline{.759} \\
            & \textbf{10} & \underline{.308} & .707  & 1.13  & .438  & .959  & .456  & \underline{.581} & 1.93  & 3.78  & 1.26  & 8.25  & 1.12  & .850 \\
      \midrule
      \multirow{5}[2]{*}{\begin{sideways}\textbf{Exponential}\end{sideways}} & \textbf{1} & \underline{.526} & .536  & .574  & .711  & 1.11  & .671  & \underline{.976} & 1.16  & 1.04  & 1.39  & 2.28  & 1.33  & 1.16 \\
            & \textbf{2} & \underline{.502} & .546  & .577  & .711  & 1.12  & .661  & \underline{.919} & 1.22  & 1.05  & 1.39  & 2.45  & 1.35  & 1.19 \\
            & \textbf{5} & \underline{.435} & .567  & .583  & .705  & 1.11  & .612  & \underline{.819} & 1.52  & 1.03  & 1.41  & 2.44  & 1.35  & 1.03 \\
            & \textbf{7} & \underline{.419} & .582  & .554  & .708  & 1.08  & .609  & \underline{.782} & 1.69  & 1.01  & 1.48  & 2.53  & 1.32  & .972 \\
            & \textbf{10} & \underline{.402} & .603  & .560  & .682  & 1.09  & .561  & \underline{.763} & 1.87  & 1.05  & 1.53  & 2.61  & 1.28  & .930 \\
      \midrule
      \multirow{5}[2]{*}{\begin{sideways}\textbf{Zipfian}\end{sideways}} & \textbf{0.6} & \underline{.493} & .543  & .630  & .720  & 1.23  & .691  & \underline{1.00} & 1.14  & 1.11  & 1.43  & 1.82  & 1.63  & 1.12 \\
            & \textbf{0.8} & \underline{.524} & .542  & .619  & .710  & 1.20  & .670  & \underline{1.00} & 1.18  & 1.09  & 1.46  & 1.92  & 1.56  & 1.08 \\
            & \textbf{1} & .601  & .631  & .648  & .735  & 1.08  & \underline{.590} & \underline{1.04} & 1.44  & 1.71  & 1.53  & 3.10  & 1.30  & 1.08 \\
            & \textbf{1.2} & \underline{.516} & .832  & 1.07  & .709  & 1.10  & .743  & \underline{.918} & 1.95  & 3.29  & 1.72  & 5.85  & 1.45  & 1.22 \\
            & \textbf{1.5} & \underline{.446} & .946  & 1.90  & .695  & 1.48  & .939  & \underline{.883} & 2.56  & 6.57  & 1.74  & 6.78  & 1.99  & 1.65 \\
      \midrule
      \multicolumn{2}{c|}{\textbf{Avg.}} & \underline{.472} & .601  & .698  & .679  & 1.11  & .629  & \underline{.882} & 1.46  & 1.49  & 1.49  & 2.92  & 1.39  & 1.07 \\
      \midrule
      \multicolumn{15}{l}{\textbf{Adversarial distributions for integer sort with mixed heavy and light keys:}} \\
      \midrule
      \multirow{6}[2]{*}{\begin{sideways}\textbf{\bitsexponential{}}\end{sideways}} & \textbf{10} & 1.11  & .833  & 1.38  & .841  & .857  & \underline{.610} & 3.30  & 2.70  & 3.31  & 1.89  & 8.03  & 1.57  & \underline{1.08} \\
            & \textbf{30} & \underline{.643} & 1.08  & 3.18  & .775  & 1.27  & .908  & 2.75  & 4.04  & 7.90  & 2.30  & 12.8  & 1.36  & \underline{1.26} \\
            & \textbf{50} & \underline{.550} & 1.20  & 4.29  & .717  & 1.60  & 1.00  & 1.85  & 4.57  & 11.9  & 2.27  & 9.47  & 1.74  & \underline{1.45} \\
            & \textbf{100} & \underline{.512} & 1.31  & 5.89  & .664  & 1.99  & 1.48  & \underline{1.42} & 4.92  & 17.8  & 2.20  & 4.96  & 2.44  & 2.03 \\
            & \textbf{300} & \underline{.616} & 1.40  & 8.22  & .606  & 2.32  & 2.02  & \underline{1.32} & 5.25  & 27.9  & 2.12  & 4.15  & 3.26  & 3.31 \\
            & \textbf{Avg.} & \underline{.659} & 1.15  & 3.91  & .716  & 1.52  & 1.11  & 1.99  & 4.19  & 10.9  & 2.15  & 7.25  & 1.97  & \underline{1.68} \\
      \bottomrule
      \end{tabular}%


\caption{\small \textbf{Running time (in seconds) on synthetic data with $\boldsymbol{n=10^9}$.}\label{tab:synthetic}
The instances of one distribution is ordered by increasing number of heavy records from top to bottom.
The fastest running time on each input instance is underlined.
``Avg.'' $=$ geometric mean.
The keys of \raduls need to be padded to multiples of 64 bits.
    We remove it from the 32-bit experiments because it is too slow after padding to 64 bits. 
}
\end{table*}%

We implemented \oursort{} in C++ and tested its performance.
\oursort{} uses ParlayLib~\cite{blelloch2020parlaylib} to support fork-join parallelism.

\myfirstparagraph{Experimental Setup.}
We run our experiments on a 96-core machine (192 hyper-threads) with 4 $\times$ 2.1 GHz
Intel Xeon Gold 6252 CPUs (each with 36MB L3 cache) and 1.5TB of main memory.
We always use \texttt{numactl -i all} to interleave the memory on all CPUs except for sequential tests.
We run each test six times and report the median of the last five runs.
All running times are in seconds.

\myparagraph{Parameter Selection.}
\oursort{} mostly follows the theoretical analysis in \cref{sec:theory}, with the minor observation that using a variable radix width $\usedbits$ (based on $n$) slightly improved the overall performance.
Hence, \oursort{} uses $\usedbits = \log_2\!\cbrt{n}$ and capped by 8 and 12 (guided by $\Theta(\!\sqrt{\log r})$, leading to the same asymptotic work bound), similar to some existing IS algorithms such as \plis{}.
The base case threshold is $\theta=2^{14}$.
\hide{
we sort $\usedbits = \min\{12, \cbrt{n}\}$ bits in each recursive call,
and
We observed that using a variable radix width $\usedbits$ (based on $n$) slightly improved the overall performance
(which is also used in some existing IS algorithms such as \plis{}),
but the value is kept below $2\sqrt{\log r}$, which matches our theory.
The number of samples is selected accordingly based on $\usedbits$.}

\myparagraph{Tested Algorithms.} We compare our \oursort{} with six \emph{baseline} parallel algorithms, including four MSD integer sorting algorithms
and two comparison-based algorithms (both are samplesorts).
We present their information in \cref{tab:baseline}.

\hide{
We tested seven algorithms, including five integer sorting algorithms (our \oursort{},
ParlayLib integer sort \plis{}~\cite{blelloch2020parlaylib}, \ipsra~\cite{axtmann2022engineering},
RegionsSort \regions{}~\cite{obeya2019theoretically}, and RADULS~\cite{kokot2018even}) and two sample sort algorithms
(ParlayLib sample sort~\cite{blelloch2020parlaylib} and \ipso~\cite{axtmann2017place}).
The notation of the algorithms is given in \cref{tab:baseline}.
}

\myparagraph{Synthetic Distributions.}
We used three \emp{standard} real-world distributions to test all algorithms:
\uniform{$\mu$} (uniform distribution with $\mu$ distinct keys),
\exponential{$\lambda$} (exponential distribution with parameter $10^{-5}\lambda$),
and \zipfian{$s$} (Zipfian distribution with parameter $s$),
We also designed a special distribution \emph{\bitsexponential{}} (\bitexp{$t$}),
which aims to simulate ``worst-case'' inputs for \oursort{}.
For \bitexp{$t$}, every bit of a key is a 0 with probability $1/t$, and 1 otherwise.
\bexp{} controls the distribution of the bitwise encoding (rather than frequency) of keys.
With a large $t$, most subproblems will have a mix of light and heavy records,
which is an adversarial case for \oursort{} due to the necessary dovetail merging for almost every subproblem.
In fact, \bexp{} is generally hard for all MSD algorithms because it makes subproblem sizes uneven.
For simplicity,
we say a distribution is \emph{heavy} (or \emph{heavier} than another)
if it contains more duplicates, and \emph{light} (\emph{lighter}) otherwise.
\hide{
Therefore, \bexp{} illustrates how the distribution of the bits (rather than frequency of keys)
affects the performance of sorting algorithms.
}

We choose five parameters for each distribution to give
various ratios of heavy keys (see \cref{tab:synthetic}).
For uniform distributions, we test $\mu = 10, 10^3, 10^5, 10^7, 10^9$.
For exponential distributions, we test $\lambda = 1, 2, 5, 7, 10$.
For Zipfian distributions, we test $s = 0.6, 0.8, 1, 1.2, 1.5$.
For \bexp{} distribution, we test $t = 10, 30, 50, 100, 300$.
We round keys to the closest integer for exponential and Zipfian distributions.
For the standard distributions, we map the keys to larger ranges (up to $2^{32}$ or $2^{64}$) to test the algorithms on the full range of $[r]$.
We use the geometric mean to compare the \emph{average} performance across multiple tests.

\subsection{Overall Performance on Synthetic Data}
We present the running time of all tested algorithms on synthetic distributions in \cref{tab:synthetic}
with different key lengths.

\myparagraph{Standard Distributions.}
With 32-bit keys, \oursort{} has the best overall performance, which is 1.27--1.47$\times$ faster than the baselines.
A heatmap for 32-bit results is also presented in \cref{fig:heatmap}.
As expected, the other integer sorting algorithms are usually competitive to \oursort{} on light distributions,
but can be much slower than \oursort{} on heavy distributions.
Take the best baseline on 32-bit keys, \plis{}, as an example. On light distributions,
It is only 2--8\% slower than \oursort{}, and is faster than comparison sorts.
However, on the heaviest several distributions in each category, \plis{} is almost always slower than the best comparison sort \ipso{},
making its overall advantage to \ipso{} small (only 4\%).
On the other hand, \oursort{} can take advantage of the heavy keys and outperforms \ipso{} in all but two tests (which we study in more details below).
Across all tests, \oursort{} can be up to 2.3$\times$ faster than \plis{}.

With 64-bit keys, all algorithms perform slower than with 32-bit keys due to the increase in data size, but the integer sorting algorithms are affected more significantly.
This matches the theory that the work bounds of these algorithms increase with the value of $r$ (see \cref{sec:theory}).
In this case, \emph{all} baseline integer sorting algorithms are slower than the fastest comparison sort \ipso{} on average,
and many of them cannot outperform \ipso{} even on some light distributions.
On the contrary, \oursort{} is still faster than both tested comparison sorts in almost all tests, and is 1.2--1.5$\times$ faster on average.

Among all tests on the standard distributions, \oursort{} \emph{performs the best on all but three tests}, where samplesort \ipso{} is better.
Two of them are \uniform{$10^3$} for both 32- and 64-bits.
The reason should be that, based on our parameter $\gamma$, the average frequency of \uniform{$10^3$} is the threshold to distinguish heavy and light keys on the root level (about $10^6$).
We tested this on purpose to show the bad case for \oursort{}, since all ``light keys'' are reasonably heavy (right around the threshold),
but they have to go through another level of recursion and have to be processed in the dovetail merging step.
This also illustrates the inherent difficulty of dealing with duplicates for parallel IS algorithms.
The last case where \oursort{} is slower than \ipso{} is \zipfian{1}, and \oursort{} is within 5\% slower.


In general, \oursort{} effectively takes advantage of different levels of heavy duplicates: it exhibits a similar trend as samplesort algorithms where the running time decreases with more duplicates. Meanwhile, as an integer sort, \oursort{} also has the benefit of being asymptotically more efficient than comparison sorts with realistic input key range $r=n^{O(1)}$.

\myparagraph{The Adversarial Distribution ({\em \emph{BExp}}).}
As mentioned above, we design the \emph{\bexp} distribution to study an adversarial case of \oursort{} and other MSD algorithms.
The distribution incurs very different sizes of MSD zones for \emph{all} subproblems, which incurs load imbalance for all MSD algorithms.
It also generates a mix of light and heavy buckets in almost all subproblems,
which greatly increases the cost of the dovetail merging and offsets the benefit of identifying heavy keys.

Almost all existing IS algorithms have poor performance on this distribution.
Interestingly, even the comparison-based sorting algorithms perform much worse when the distribution is very skewed.
Among the MSD algorithms, \ipsra{} and \raduls{} are significantly affected and can be up to 10$\times$ slower than the others.
\oursort{}, \plis{}, and \regions{} are lightly affected on 32-bit keys,
but in general can still be within 2$\times$ slower compared to their performance on standard distributions.
A possible explanation is that all the three algorithms use similar work-stealing scheduler (\texttt{cilk} for \regions{}, ParlayLib scheduler for \oursort{} and \plis{}), which supports dynamic load balancing.
This matches our scalability tests (reported in \ifconference{the full paper~\cite{dong2024parallelfull}}\iffullversion{in \cref{sec:app-scalability}}), where \ipsra{} and \raduls{} do not scale beyond 24 threads on this distribution.
\regions{} seems to be particularly resistant to the \bexp{} distribution.
One reason is that \regions{} is \emph{unstable}, and records with the same key do not need to be put in a specific order.
This sacrifice on stability has a side-effect to benefit from the highly skewed distributions.

Overall, \oursort{} still achieves the best performance  on 32-bit keys,
and is still faster than comparison sorts in most cases.
We believe this advantage comes from the removal of heavy keys from recursive problems, and thus alleviates the load-imbalance issue.
In addition, the performance of \oursort{} is also significantly improved by using our proposed \dtmerge{} algorithm.
As we will show in \cref{sec:exp-study}, as expected, on \bexp{} distributions,
the merging step is very expensive, which can take 50\% of the entire running time if we simply use parallel merge.
Our new algorithm accelerates this step by 1.7--2.8$\times$, making the overall running time much lower.

\begin{table}[t]
\small
\centering
\begin{tabular}{@{}>{\bf}rr@{}r|rr@{  }rr|r@{  }r}

      \multicolumn{2}{@{}c@{}}{\multirow{2}[1]{*}{\textbf{Instances}}} & \multicolumn{1}{c|}{\multirow{2}[1]{*}{\boldmath{}\textbf{$\boldsymbol{n}$}\unboldmath{}}} & \multicolumn{4}{c|}{\textbf{Integer}} & \multicolumn{2}{@{  }c@{  }}{\textbf{Comparison}} \\
      &      &       & \multicolumn{1}{c}{\textbf{Ours}} & \multicolumn{1}{c}{\textbf{PLIS}} & \multicolumn{1}{@{}c@{}}{\boldmath{}\textbf{IPS$^2$Ra}\unboldmath{}} & \multicolumn{1}{c|}{\textbf{RS}} & \multicolumn{1}{c@{}}{\textbf{PLSS}} & \multicolumn{1}{c}{\boldmath{}\textbf{IPS$^4$o}\unboldmath{}} \\
      \midrule
      \multicolumn{8}{c}{\textbf{Graph transpose}} \\
      \midrule
      &\multicolumn{1}{@{  }c}{\textbf{LJ}} & \textbf{69.0M} & \underline{.043}  & .043 & s.g.  & .065  & .080  & .159 \\
      &\multicolumn{1}{@{  }c}{\textbf{TW}} & \textbf{1.47B} & \underline{.888} & .942  & 3.24  & 1.05  & 1.57  & .891 \\
      &\multicolumn{1}{@{  }c}{\textbf{CM}} & \textbf{1.61B} & \underline{.782} & .945  & 1.41  & 1.07  & 1.84  & 1.12 \\
      &\multicolumn{1}{@{  }c}{\textbf{SD}} & \textbf{2.04B} & \underline{1.10} & 1.29  & 2.87  & 1.34  & 2.08  & 1.23 \\
      &\multicolumn{1}{@{  }c}{\textbf{CW}} & \textbf{42.6B} & 28.5  & 37.0  & 32.5  & s.g.  & 60.1  & \underline{24.6} \\
      &\multicolumn{2}{r|}{\textbf{Avg.}} & \underline{.985} & 1.13  & -     & -     & 1.96  & 1.37 \\
      \midrule
      \midrule
      \multicolumn{8}{c}{\textbf{Morton order}} \\
      \midrule
      \multirow{4}{*}{\begin{sideways}Real\par World\end{sideways}} &\multicolumn{1}{@{  }c}{\textbf{GL}} & \textbf{24.9M} & .026  & .028  & .259  & \underline{.024} & .028  & .171 \\
      &\multicolumn{1}{@{  }c}{\textbf{CM}} & \textbf{321M} & .184  & \underline{.178} & .343  & .209  & .327  & .338 \\
      &\multicolumn{1}{@{  }c}{\textbf{OSM}} & \textbf{2.77B} & 2.32  & 2.39  & 3.65  & s.g.  & 2.73  & \underline{1.53} \\
      &\multicolumn{2}{r|}{\textbf{Avg.}} & \underline{.223} & .227  & .687  & -     & .293  & .445 \\
      \midrule
      \multirow{4}{*}{\begin{sideways}Varden~\cite{gan2017hardness}\end{sideways}} &\multicolumn{1}{l}{\textbf{SS2d}} & \textbf{1B} & \underline{.498} & .557  & .662  & .634  & 1.27  & .775 \\
      &\multicolumn{1}{l}{\textbf{SS3d}} & \textbf{1B} & \underline{.512} & .568  & .778  & .611  & 1.12  & .754 \\
      &\multicolumn{1}{l}{\textbf{SS2d'}} & \textbf{2B} & \underline{.973} & 1.16  & 1.27  & 1.17  & 2.44  & 1.39 \\
      &\multicolumn{1}{l}{\textbf{SS2d'}} & \textbf{2B} & \underline{.990} & 1.60  & 2.86  & 1.17  & 2.30  & 1.96 \\
      &\multicolumn{2}{r|}{\textbf{Avg.}} & \underline{.704} & .875  & 1.17  & .854  & 1.68  & 1.12 \\
      \end{tabular}%


\caption{\small \textbf{Running time (in seconds) on multiple applications.}\label{tab:application}
We use 32-bit keys and 32-bit values.
See \cref{sec:application} for the detailed explanation on the keys and values.
The fastest running time on each instance is underlined.
``$n$'' $=$ input sizes.
``s.g.'' $=$ segmentation fault.
``Avg.'' $=$ the geometric mean on instances of the same application.
``-'' $=$ not applicable.\vspace{1em}}
\end{table}%

\subsection{Applications}\label{sec:application}
Integer sort is widely used in practice.
We choose two representative applications on graph and geometric processing to evaluate the performance of the tested algorithms.

\myparagraph{Graph Transpose.}
Given a directed graph $G=(V, E)$, the graph transpose problem is to generate $G^\matrixtrans=(V, E^\matrixtrans)$ where $E^\matrixtrans=\{(v, u): (u, v) \in E\}$.
Here we assume the edges are stored using the standard format of compressed sparse row.
To compute the transposed graph, we only need to stably integer sort all edges using the key as the second vertex $v$; here the vertices with larger degrees are the ``heavy keys''.
Graph transpose is widely used as a subroutine in many graph algorithms~\cite{blelloch2016parallelism,gbbs2021,ji2018ispan,slota2014bfs}.

We tested five real-world directed graphs,
including soc-LiveJournal (LJ)~\cite{backstrom2006group},
twitter (TW)~\cite{kwak2010twitter}, Cosmo50 (CM)~\cite{cosmo50,wang2021geograph},
sd\_arc (SD)~\cite{webgraph}, and clueweb (CW)~\cite{webgraph}.
The largest graph has 42.6 billion directed edges.
We save each edge as a (32-bit, 32-bit) integer pair since all vertex ids are within 32 bits.
On social networks (LJ and TW) and web graphs (SD and CW),
the degree distributions are more skewed (i.e., more duplicate keys).
CM is a $k$-NN graph, and the degrees are more evenly-distributed.
We present our results in \cref{tab:application}.

Compared to the best baseline (\plis{}), \oursort{} is always faster by up to 1.3$\times$.
The samplesort \ipso{} performs the best on CW,
but performs much slower on small graphs.  
Overall, \oursort{} has the best performance across all graphs.

\myparagraph{Morton Sort.}
For a $d$-dimensional point,
a z-value is calculated by interleaving the binary representations of each coordinate.
The process to sort the z-values of a set of points is called Morton sort,
which orders multidimensional data in one dimension while preserving the locality of data points.
It is widely used in geometry, graphics, and machine learning applications~\cite{blelloch2016just,wang2020theoretically,blelloch2022parallel,gu2013efficient,gu2015ray,li2012kann}.

We tested seven datasets for Morton sort and list the result in \cref{tab:application}.
Three are real-world datasets, including GeoLife (GL)~\cite{geolife,wang2021geograph},
Cosmo50 (CM)~\cite{cosmo50,wang2021geograph} and OpenStreetMap (OSM)~\cite{roadgraph}.
The other four are synthetic datasets by the widely-used generator Varden~\cite{gan2017hardness} that creates points with varying densities.
On the synthetic distribution with varying densities, \oursort{} outperforms all other baselines by 1.2--3.4$\times$.
On the three datasets with real-world points, existing integer sorts (\regions{} or \plis{}) achieved the best performance on GL and CM,
while \oursort{} is competitive (3--17\% slower). On OSM, none of the integer sorts outperform the samplesort \ipso{}.
Interestingly, \oursort{} still achieved the best average performance on the three datasets (\plis{} is also close).
This may be due to the reason that \oursort{} \emph{combines the benefit for both comparison and integer sorts},
and hence \emph{does not have a very bad worst-case performance}.
\hide{
Overall, our algorithm achieves the best performance on overage.
which is 13\% faster than the second best algorithm (\plis).
Our integer algorithm is competitive with other sorting algorithms except for OSM.
}

\begin{figure*}
  \small
  \centering
    \begin{minipage}{.66\columnwidth}\centering
      \includegraphics[width=\textwidth]{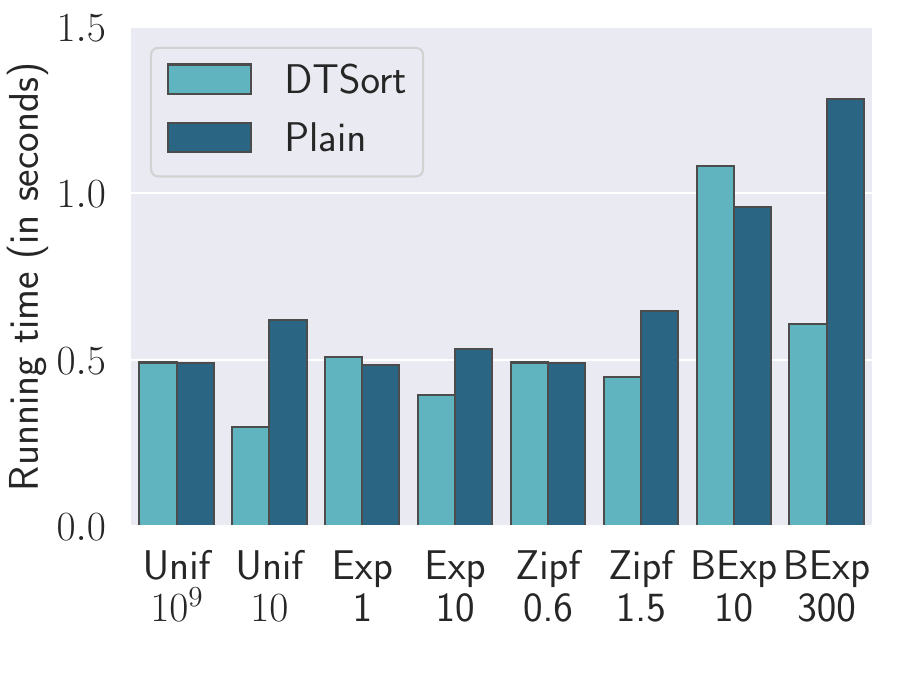}\\\vspace{-.8em}(a)
    \end{minipage}
    \begin{minipage}{.66\columnwidth}\centering
          \includegraphics[width=\textwidth]{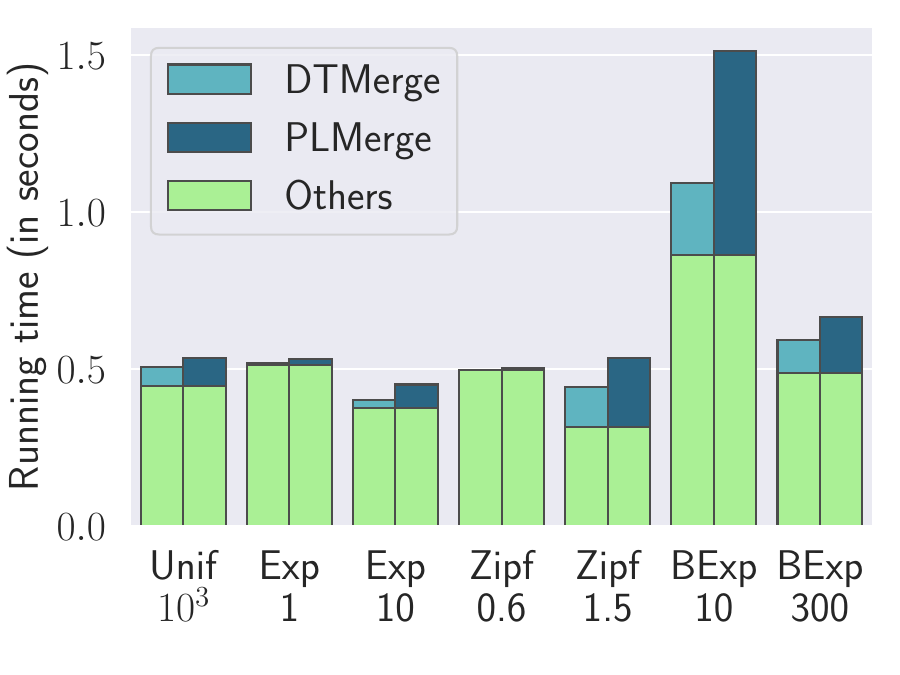}\\\vspace{-.8em}(c)
    \end{minipage}
    \begin{minipage}{.66\columnwidth}\centering
          \includegraphics[width=\textwidth]{output/scaling/32bit-zipfian-0.8.pdf}\\\vspace{-.3em}(e)
    \end{minipage}

    \begin{minipage}{.66\columnwidth}\centering
      \includegraphics[width=\textwidth]{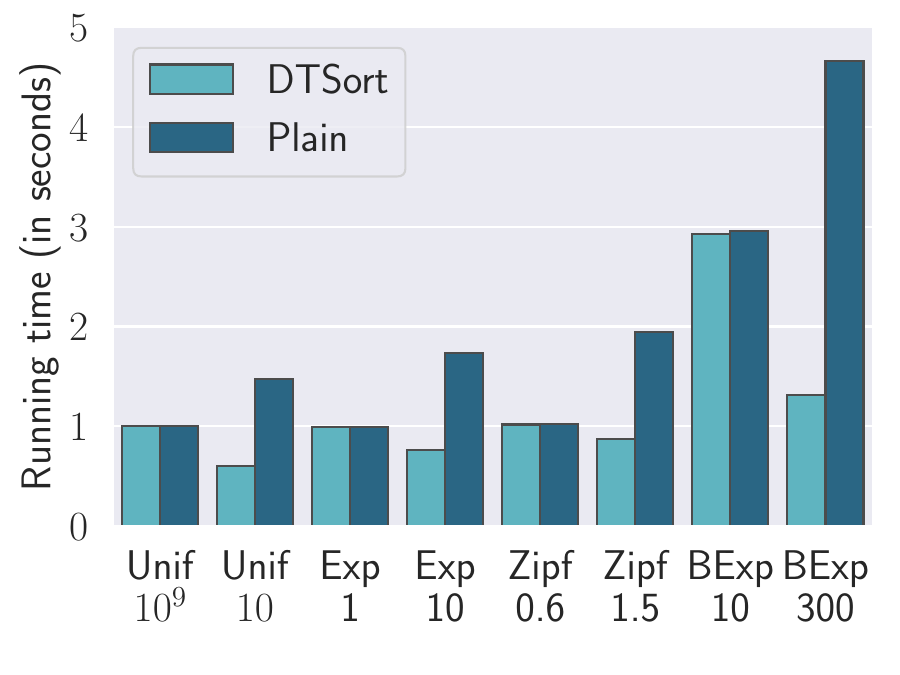}\\\vspace{-.8em}(b)
    \end{minipage}
    \begin{minipage}{.66\columnwidth}\centering
      \includegraphics[width=\textwidth]{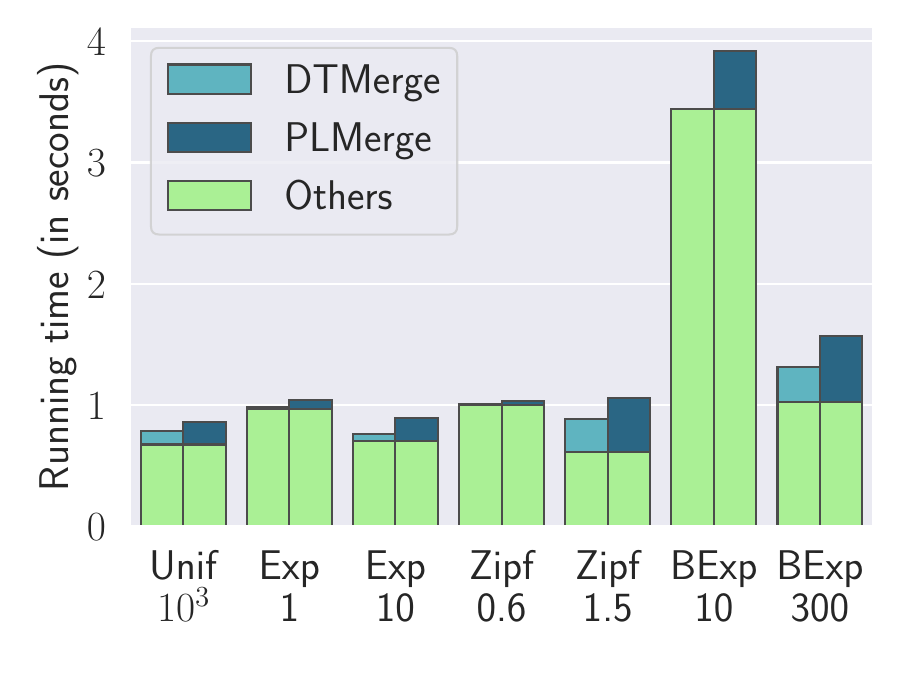}\\\vspace{-.8em}(d)
    \end{minipage}\begin{minipage}{.66\columnwidth}\centering
      \includegraphics[width=\textwidth]{output/increasing_input_size/32bit-zipfian-0.8.pdf}\\\vspace{-.3em}(f)
    \end{minipage}
  \caption{
  (a) and (b): Analysis for the performance of heavy-key detection.
  Numbers are running time (lower is better) with or without heavy-key detection. (a) is for 32-bit keys and (b) is for 64-bit keys.
  (c) and (d): Analysis for the performance of dovetail merging.
  Numbers are running time (lower is better) using our dovetail merging algorithm or a baseline merging algorithm. 
  (c) is for 32-bit keys and (d) is for 64-bit keys.
  (e) and (f): Scalability (higher is better) with varying number of threads and running time (lower is better) with varying input sizes on 32-bit key and 32-bit value pairs on one instance: \zipfian{0.8}. Full analysis is given in \ifconference{the full paper~\cite{dong2024parallelfull}}\iffullversion{\cref{sec:app-scalability}}.
  Discussions are in \cref{sec:exp-study}.}
  \label{fig:scale}
\end{figure*}

\subsection{Performance Study}\label{sec:exp-study}
\myfirstparagraph{Heavy Keys Detection Evaluation.}
To understand the performance gain by detecting heavy keys and processing them separately,
we compare our implementation with (labeled as \oursort{}) and without (labeled as \impname{Plain}) heavy keys detection in \cref{fig:scale} (a) and (b).
We select eight representative distributions, which are the lightest and heaviest cases for each distribution respectively.

As discussed in \cref{sec:oursort}, the overhead of \oursort{} compared to \plain{} includes
sampling to detect heavy keys, and merging heavy keys with light keys.
On 32-bit inputs, when there are no or fewer heavy keys, \oursort{} is comparable with \plain{},
with at most 12\% performance loss on \bitexp{10}.
It indicates that the overhead of sampling is negligible compared to other steps.
With heavy duplicate keys, since heavy keys are gathered in separate buckets
and do not need to be sorted in later recursions,
\oursort{} outperforms \plain{} by up to 2.12$\times$ (on \bitexp{300}).
Overall, \oursort{} outperforms \plain{} by 25\% on average.
On 64-bit input, the performance improvement by using heavy keys detection becomes larger.
\oursort{} consistently outperforms \plain{}, with a 1.50$\times$ speedup on average.

\myparagraph{Dovetail Merge Evaluation.} In this section, we present an in-depth study of the dovetail merging step
that implements the heavy-light separation idea.
As mentioned in \cref{sec:merge}, our baseline solution (parallel merge), is correct and even asymptotically better, but it requires to move all records between the input array and a temporary array.
To achieve the best performance, we utilized the special property of this problem (i.e., very few heavy buckets) and proposed our \dtmerge{} algorithm in \cref{sec:merge}.
We select seven representative distributions to evaluate the merge performance.
For uniform distribution, we only use \uniform{$10^3$}, since no merge is needed in the lightest and heaviest cases.
For other distributions, we select the lightest and heaviest cases.

To illustrate the effectiveness of our \dtmerge{} approach, we compare it against a standard parallel merge in ParlayLib~\cite{blelloch2020parlaylib} (referred to as \plmerge{}) as the baseline.
When reporting the running time, we also run the algorithms \emph{without} the merging step in all recursive calls. This will not give a correct sorting result, but indicates the running time of the other steps in \oursort{}.
We present the result for both 32 bits and 64 bits in \cref{fig:scale} (c) and (d).
Since they generally have similar trends, we discuss the experiments based on 32 bits.

On lighter distributions (e.g., \exponential{1} and \zipfian{0.6}), the overall improvement is marginal given that the dovetail merging step is fast anyway.
However, on heavier (more skewed) distributions,
this step can take a significant portion of the overall time if we use \plmerge{}, which is 17\% on \exponential{10} and 42\% on \zipfian{1.5}.
With our \dtmerge{}, this step only takes 6\% in \exponential{10} and 28\% for \zipfian{1.5}, improving the overall performance for both by about 20\%.

For \bexp{}, since almost all subproblems have a mix of heavy and light keys, the merging time is significant.
Using \plmerge{}, the merging step takes 27--43\% running time.
Our \dtmerge{} improved the performance of the merging step by up to 2.8$\times$,
leading to 1.1--1.38$\times$ faster sorting time.

\myparagraph{Scalability Tests.}
\oursort{} has good scalability to both the number of processors and input sizes.
Due to page limit, we provide results on one distribution (\zipfian{0.8}) as an example in \cref{fig:scale} (e) and (f), and give the full discussions in \ifconference{the full paper~\cite{dong2024parallelfull}}\iffullversion{\cref{sec:app-scalability}}.

\section{Conclusion}
We study the theory and practice of parallel integer sort.
In theory, we show that a class of widely-used practical MSD integer sorting algorithms have $O(n\sqrt{\log r})$ work and low span.
The bounds are asymptotically lower than comparison sort for a large key range $r=n^{o(\log n)}$,
and provide a solid theoretical foundation for explaining the high performance of parallel integer sort.
In practice, we propose \oursort{} to overcome a common challenge in existing MSD algorithms:
the difficulty to detect and take advantage of duplicate keys.
\oursort{} combines algorithmic insights from both integer and comparison sorting algorithms,
and uses a special dovetail merging step to make it more efficient.
\oursort{} achieves the same $O(n\sqrt{\log r})$ work bound as other MSD integer sorting algorithms,
and has the same advantage of processing heavy duplicates efficiently as samplesort algorithms.
In our experiments, \oursort{} achieves competitive or better performance than state-of-the-art parallel integer and comparison sorting algorithms
on various synthetic and real-world datasets.  As a future work, we consider further improving the performance of parallel sorting algorithm by accelerating the distributing step, inspired by some existing work~\cite{obeya2019theoretically,axtmann2022engineering,gu2021parallel}.

\bibliographystyle{ACM-Reference-Format}
\balance

\iffullversion{
\appendix

\section{Determinacy Race}

A \emp{determinacy race} is when two logically parallel operations access the same memory location and at least one of them is a write~\cite{CLRS}.
We say an algorithm is \emph{race-free} if it has no determinacy race.
Note that this requirement is stronger than data race~\cite{gharachorloo1990memory,adve1990weak}.
A race-free algorithm is (internally) deterministic~\cite{blelloch2012internally},
and has many advantages, 
including ease of reasoning about the code, verifying correctness, debugging, and analyzing the performance.
For \oursort{}, all operations in the algorithm are deterministic once we fix the random seed,
which is a useful feature in practice.

\section{More Details about Counting Sort}\label{sec:app-counting}

Counting sort is a commonly-used primitive in parallel sorting algorithms to distribute records into (a small number of) buckets.
It takes an input of $n'$ records with a function to map each record to an integer (the \emp{bucket id}) in $[r']$ ($r'\le n'$).
The goal is to sort the records by the bucket ids.
The classic stable counting sort~\cite{vishkin2010thinking}, which is widely used in practice, has $O(n')$ work and $O(r'+\log n')$ span, and can be made I/O-efficient~\cite{blelloch2010low}.
Here we do not use $n$ since counting sort is used in the recursive structure of many sorting algorithms and thus the subproblem size $n'$ can be smaller than $n$.
The algorithm partitions the input array into $l=\Theta(n'/r')$ subarrays each with size $\Theta(r')$.
For each subarray, we sequentially count the appearances of each key, and in total we store $l\cdot r'=\Theta(n)$ such counts in a matrix (the counting matrix).
The total count of each key can be computed by a reduce on the counts for all subarrays.
We then apply a prefix-sum on the counting matrix to get the offset of each bucket in each subarray, and finally distribute the records to their destinations in the buckets.
In practice, to make it faster, we pick the number of subarrays smaller than $m/r'$ so the counting matrix is smaller (and all operations on it can be faster).
Under this premise, a smaller number of subarrays is preferred to improve the performance, since it makes the counting matrix small to fit in cache~\cite{dong2023high}.

There exists a randomized algorithm for unstable counting sort on $n'=\Omega(\log^2n)$ records with $r'=O(n'\log^{O(1)}n)$ buckets in $O(n'+r')$ work and $O(\log n)$ span \whp{} in $n$~\cite{RR89}.
However, this approach is complicated and is less frequently used in practice due to some overhead.
The reason is that this algorithm requires allocating an $O(n)$-size array but with a large constant hidden in the Big-O, and uses many random accesses to distribute the records to the associated buckets which is I/O-unfriendly.
Finally, it also needs to pack the large array.
Putting all pieces together makes the algorithm less practical than the stable counting sort mentioned above.

\section{More Experimental Results}\label{sec:app-scalability}

\myparagraph{Parallel Scalability.} We test all algorithms using various numbers of threads (from 1 to 192) to show how they scale with the number of threads.
For each distribution, we choose two representative tests (a heavy one and a light one), and present our results in \cref{fig:scalability-uniform-1e7-32,fig:scalability-uniform-1e3-32,fig:scalability-exponential-2e-5-32,fig:scalability-exponential-7e-5-32,fig:scalability-Zipfian-0.7-32,fig:scalability-Zipfian-1.2-32,fig:scalability-bitsexp-30-32,fig:scalability-bitsexp-100-32,fig:scalability-uniform-1e7-64,fig:scalability-uniform-1e3-64,fig:scalability-exponential-2e-5-64,fig:scalability-exponential-7e-5-64,fig:scalability-Zipfian-0.7-64,fig:scalability-Zipfian-1.2-64,fig:scalability-bitsexp-30-64,fig:scalability-bitsexp-100-64}.
On the three standard distributions, most algorithms present favorably good speedups.
Our algorithm achieves the top-tier (almost linear) speedup, with up to 72$\times$ on 192 hyper-threads.

On the \bexp{} distribution, the self-relative speedups are slightly worst for all algorithms.
In particular, \ipsra{}, \ipso{} and \raduls{} do not scale well (or become even worse) with increasing number of threads.
As discussed in the main body, this is mainly because \bexp{} creates subproblems with very different sizes,
and thus is challenging to be handled by integer sorts due to load imbalance.
In most cases, \oursort{}, \regions{}, \plss{} and \plis{} still scale to at least 96 cores (\regions{}'s performance drops slightly with hyperthreading).
This may be due to the use of a work-stealing scheduler in these implementations, which enables dynamic load balancing.
Our algorithm still achieves almost-linear scalability, and has the best speedups in most cases.

\myparagraph{Input Size Scalability.} We test all algorithms with various input sizes for both 32-bit and 64-bit keys.
For each distribution, we choose two representative tests (a heavy one and a light one), and present our results in \cref{fig:input-size-uniform-1e7-32,fig:input-size-uniform-1e3-32,fig:input-size-exponential-2e-5-32,fig:input-size-exponential-7e-5-32,fig:input-size-Zipfian-0.8-32,fig:input-size-Zipfian-1.2-32,fig:input-size-bitsexp-30-32,fig:input-size-bitsexp-100-32,fig:input-size-uniform-1e7-64,fig:input-size-uniform-1e3-64,fig:input-size-exponential-2e-5-64,fig:input-size-exponential-7e-5-64,fig:input-size-Zipfian-0.8-64,fig:input-size-Zipfian-1.2-64,fig:input-size-bitsexp-30-64,fig:input-size-bitsexp-100-64}.
We remove \raduls{} from these figures since it has a serious scale-down issue on some of the distributions (needs more than two seconds with $n=10^7$).
We measure the input size scalability starting from $n=10^7$,
because on an even smaller size, most algorithms only take less than ten milliseconds.
Overall, our algorithm \oursort{}, \plis, and \plss scale well in a large range ($n=10^7$ to $2\times10^9$).
and our algorithm achieves the best performance in most cases.
\regions{}, \ipsra, and \ipso do not perform well when the input sizes are relatively small ($n \leq 10^8$),
but they also achieve good performance when the input sizes are sufficiently large.

\begin{figure*}[t]
{\Large Self-speedup with Varying Thread Counts ($n=10^9$, 32-bit keys and 32-bit values). Higher is better.}\\
Uniform Distribution\\
\begin{minipage}{.9\columnwidth}
  \centering
  \includegraphics[width=.8\columnwidth]{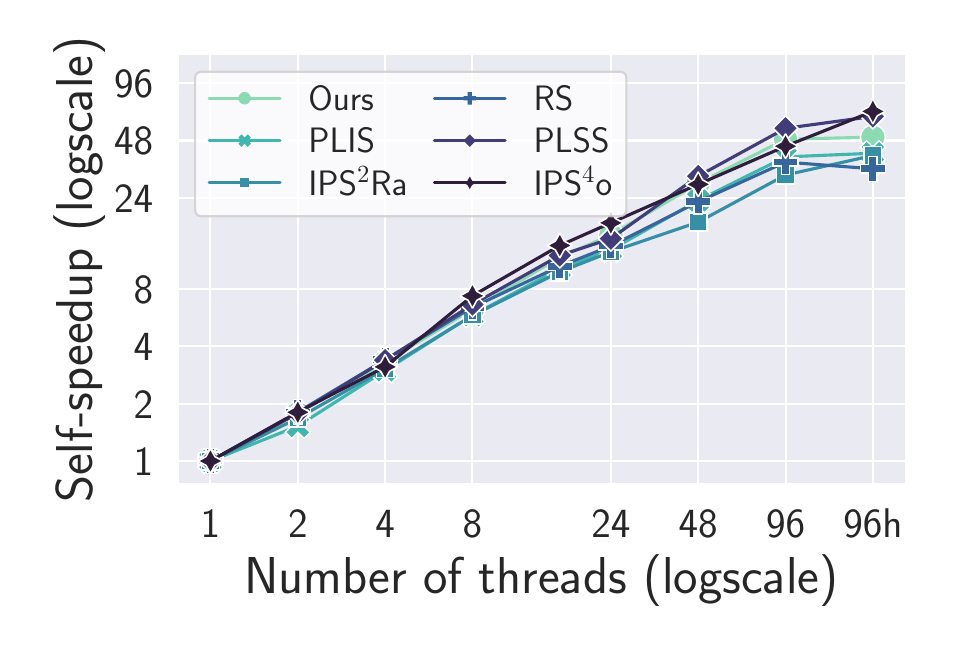}
  \vspace{-.5em}
  \caption{\textbf{Self-speedup with varying thread counts of all tested implementations on \uniform-$\boldsymbol{10^7}$.}}\label{fig:scalability-uniform-1e7-32}
\end{minipage}\hfil
\begin{minipage}{.9\columnwidth}
  \centering
  \includegraphics[width=.8\columnwidth]{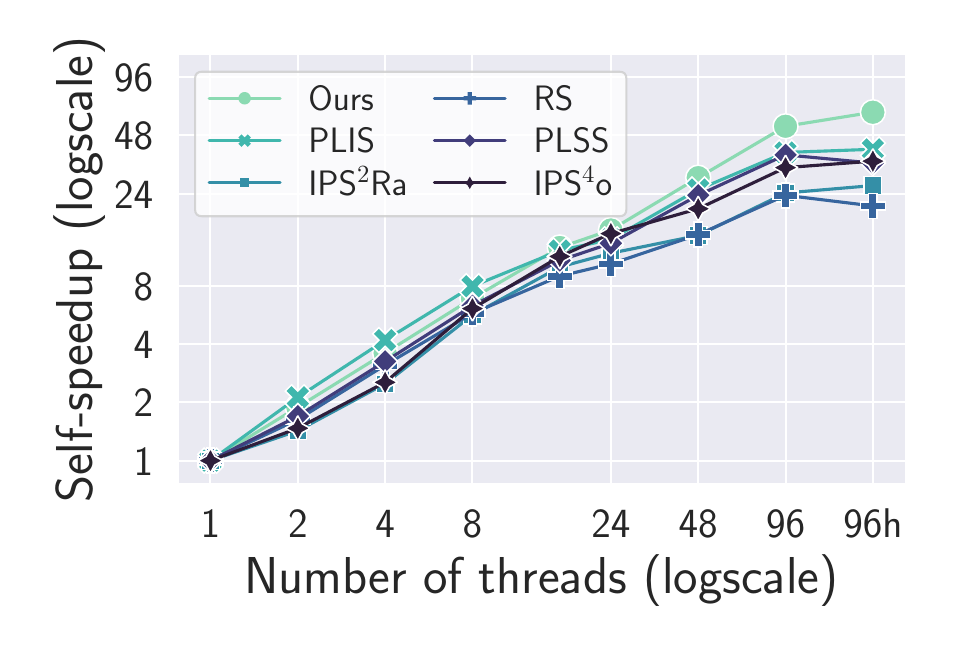}
  \vspace{-.5em}
  \caption{\textbf{Self-speedup with varying thread counts of all tested implementations on \uniform-$\boldsymbol{10^3}$.}}\label{fig:scalability-uniform-1e3-32}
\end{minipage}

\bigskip
Exponential Distribution\\
\begin{minipage}{.9\columnwidth}
  \centering
  \includegraphics[width=.8\columnwidth]{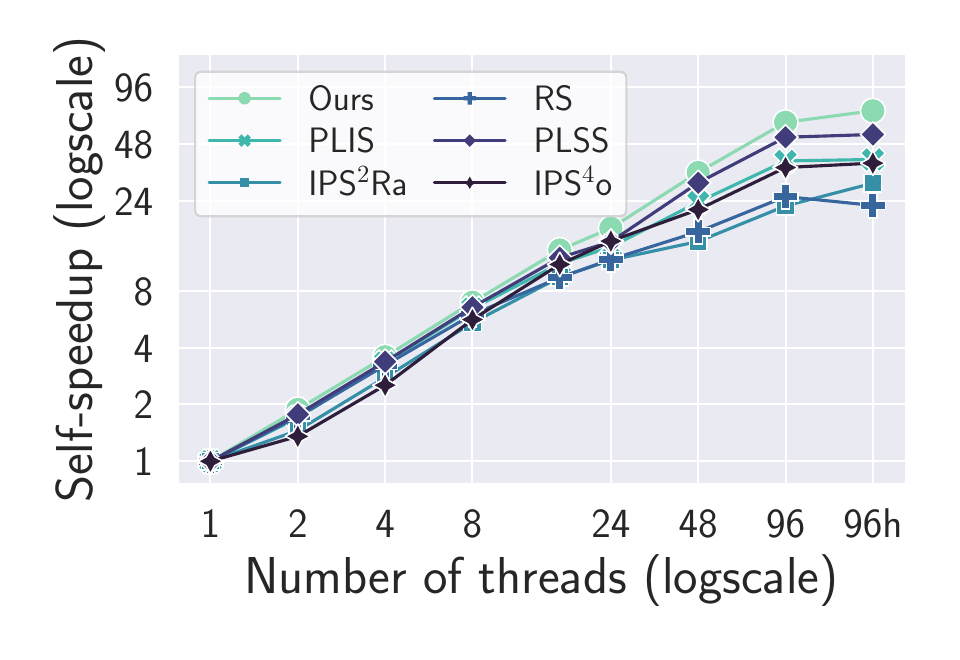}
  \vspace{-.5em}
  \caption{\textbf{Self-speedup with varying thread counts of all tested implementations on \exponential-$\boldsymbol{2}$.}}\label{fig:scalability-exponential-2e-5-32}
\end{minipage}\hfil
\begin{minipage}{.9\columnwidth}
  \centering
  \includegraphics[width=.8\columnwidth]{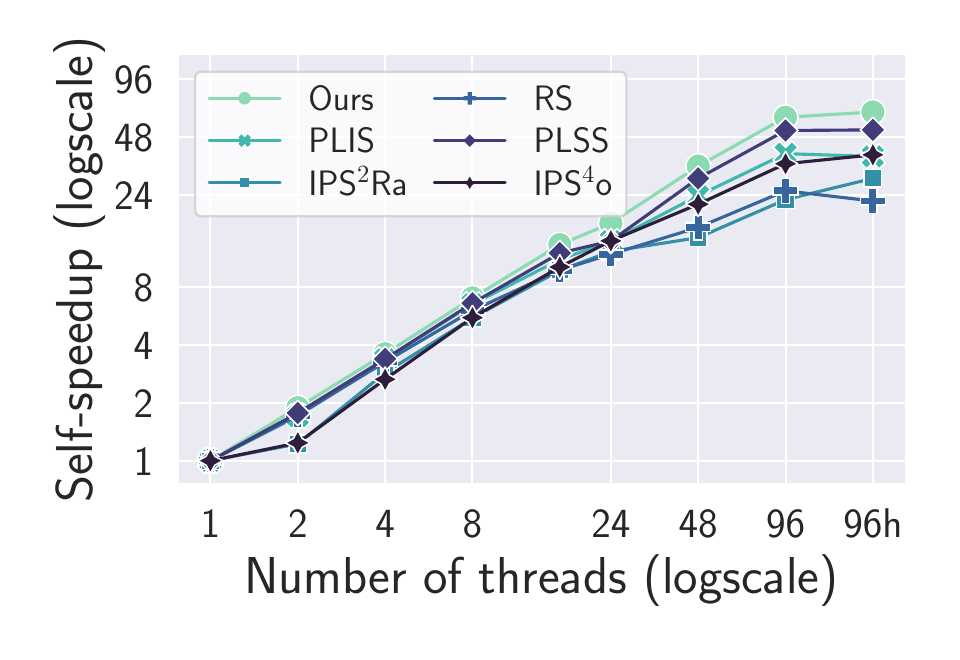}
  \vspace{-.5em}
  \caption{\textbf{Self-speedup with varying thread counts of all tested implementations on \exponential-$\boldsymbol{7}$.}}\label{fig:scalability-exponential-7e-5-32}
\end{minipage}

\bigskip
Zipfian Distribution\\
\begin{minipage}{.9\columnwidth}
  \centering
  \includegraphics[width=.8\columnwidth]{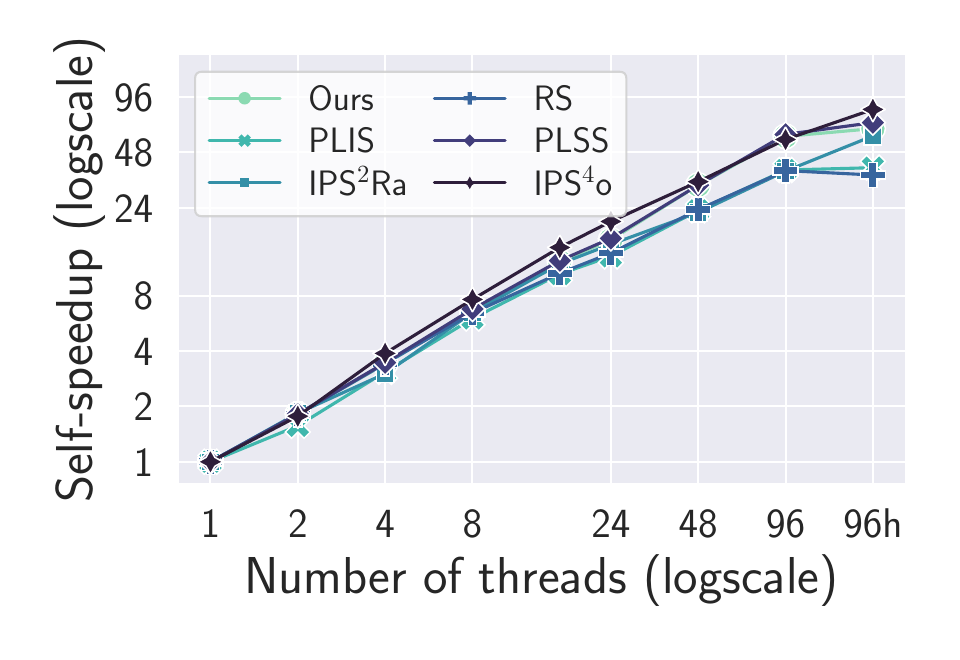}
  \vspace{-.5em}
  \caption{\textbf{Self-speedup with varying thread counts of all tested implementations on \zipfian-$\boldsymbol{0.8}$.}}\label{fig:scalability-Zipfian-0.7-32}
\end{minipage}\hfil
\begin{minipage}{.9\columnwidth}
  \centering
  \includegraphics[width=.8\columnwidth]{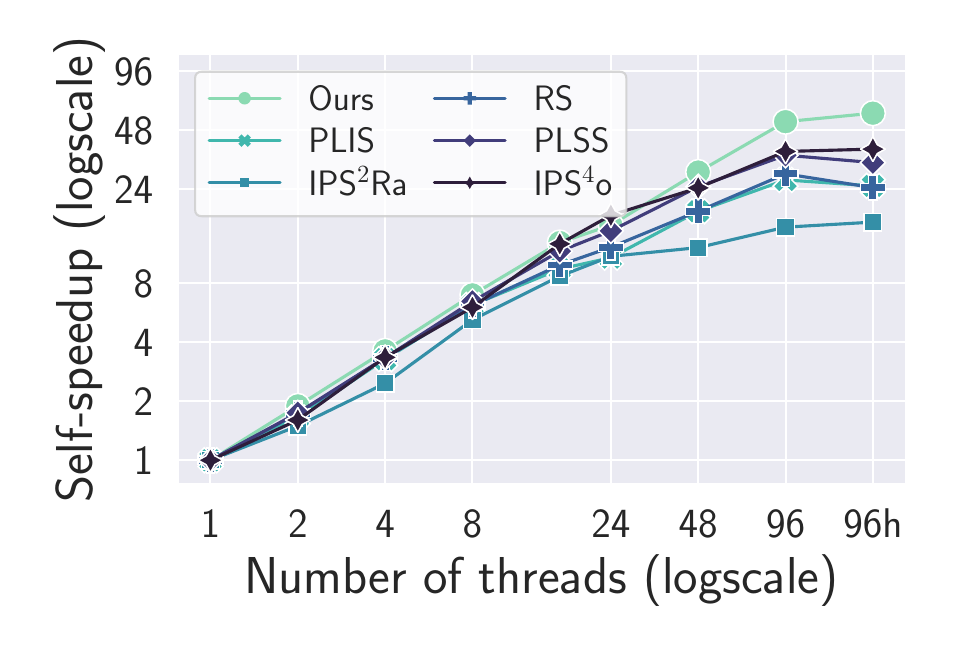}
  \vspace{-.5em}
  \caption{\textbf{Self-speedup with varying thread counts of all tested implementations on \zipfian-$\boldsymbol{1.2}$.}}\label{fig:scalability-Zipfian-1.2-32}
\end{minipage}

\bigskip
\bexp{} Distribution\\
\begin{minipage}{.9\columnwidth}
  \centering
  \includegraphics[width=.8\columnwidth]{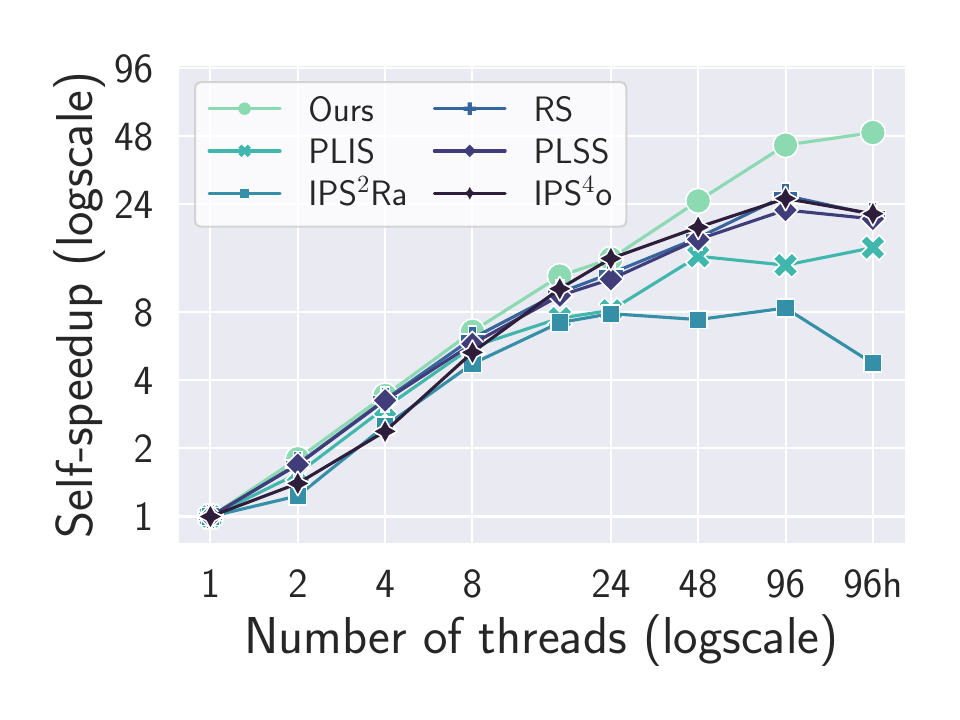}
  \vspace{-.5em}
  \caption{\textbf{Self-speedup with varying thread counts of all tested implementations on \bitexp-$\boldsymbol{30}$.}}\label{fig:scalability-bitsexp-30-32}
\end{minipage}\hfil
\begin{minipage}{.9\columnwidth}
  \centering
  \includegraphics[width=.8\columnwidth]{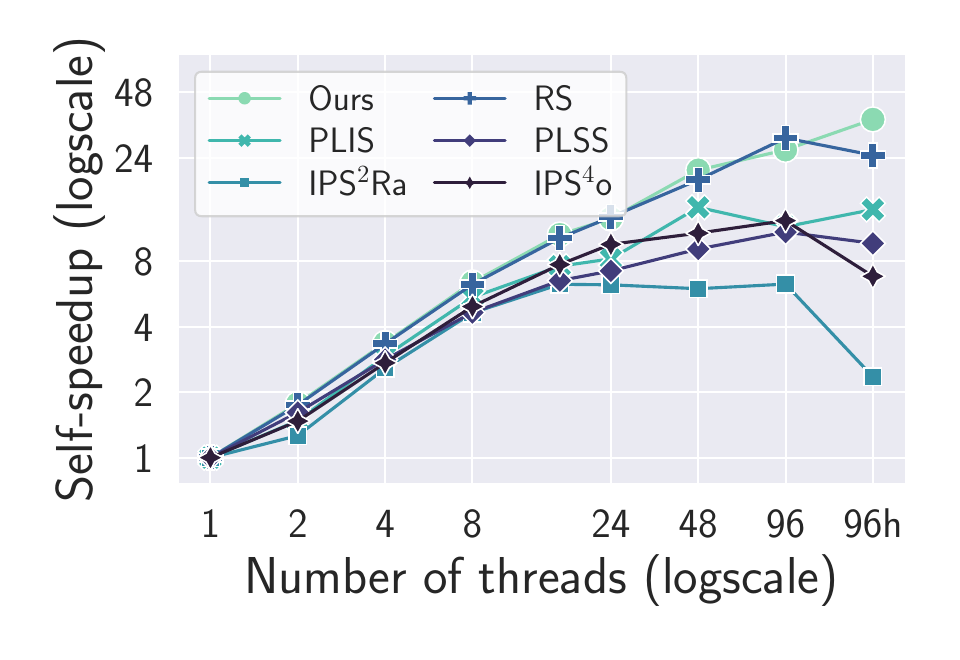}
  \vspace{-.5em}
  \caption{\textbf{Self-speedup with varying thread counts of all tested implementations on \bitexp-$\boldsymbol{100}$.}}\label{fig:scalability-bitsexp-100-32}
\end{minipage}
\end{figure*}

\begin{figure*}[t]
    {\Large Self-speedup with Varying Thread Counts ($n=10^9$, 64-bit keys and 64-bit values). Higher is better.}\\
    Uniform Distribution\\
    \begin{minipage}{.9\columnwidth}
      \centering
      \includegraphics[width=.8\columnwidth]{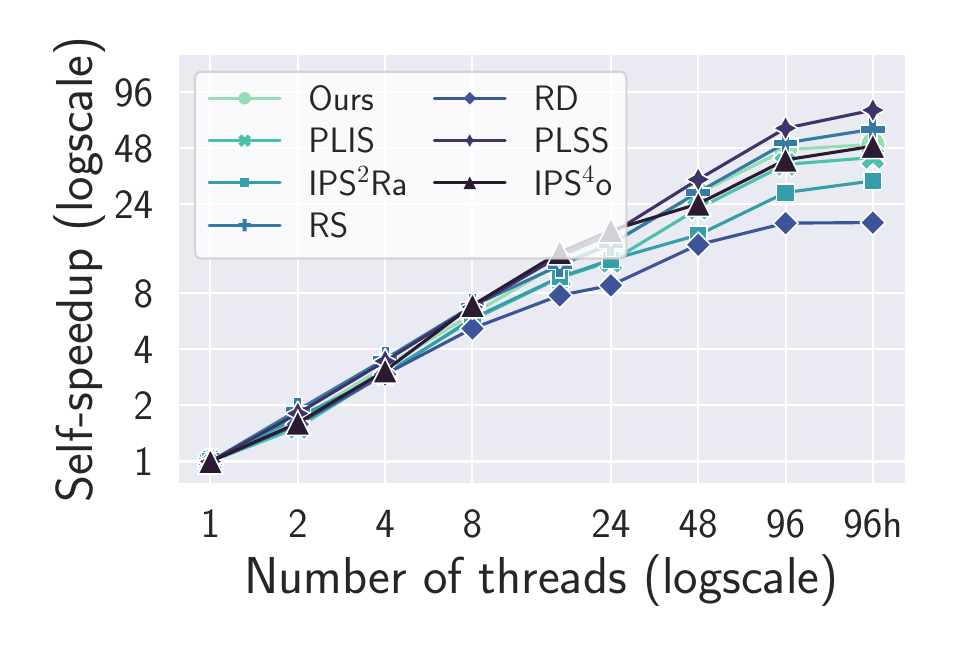}
  \vspace{-.5em}
      \caption{\textbf{Self-speedup with varying thread counts of all tested implementations on \uniform-$\boldsymbol{10^7}$.}}\label{fig:scalability-uniform-1e7-64}
    \end{minipage}\hfil
    \begin{minipage}{.9\columnwidth}
      \centering
      \includegraphics[width=.8\columnwidth]{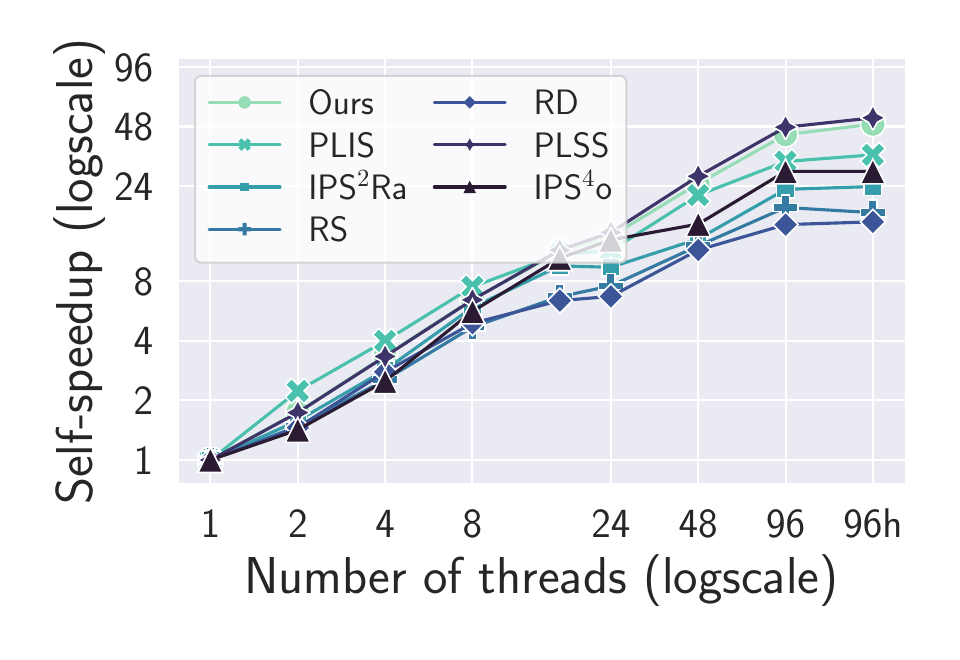}
  \vspace{-.5em}
      \caption{\textbf{Self-speedup with varying thread counts of all tested implementations on \uniform-$\boldsymbol{10^3}$.}}\label{fig:scalability-uniform-1e3-64}
    \end{minipage}

\bigskip
    Exponential Distribution\\
    \begin{minipage}{.9\columnwidth}
      \centering
      \includegraphics[width=.8\columnwidth]{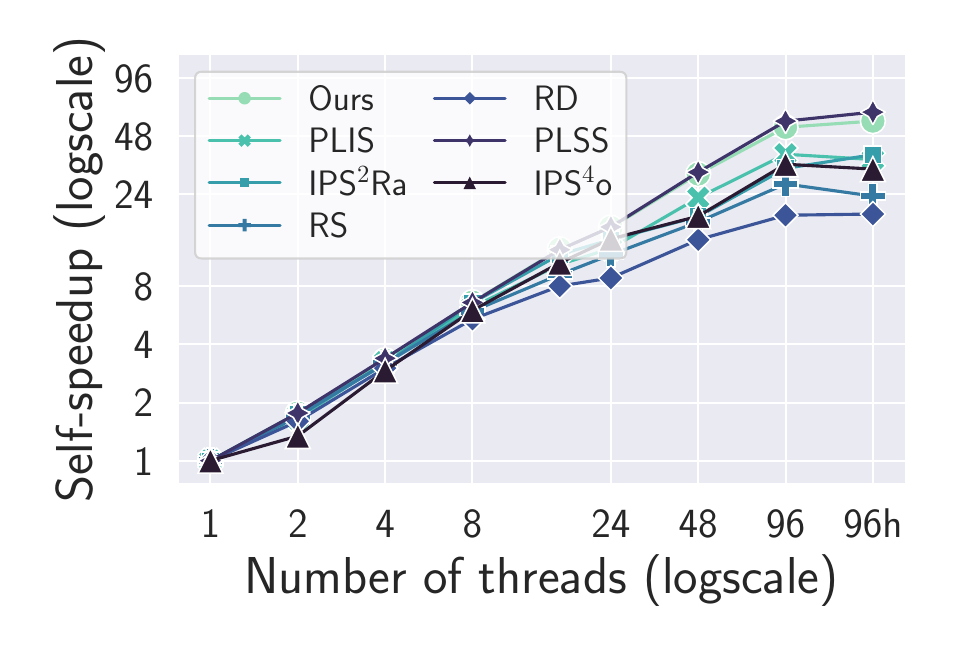}
  \vspace{-.5em}
      \caption{\textbf{Self-speedup with varying thread counts of all tested implementations on \exponential-$\boldsymbol{2}$.}}\label{fig:scalability-exponential-2e-5-64}
    \end{minipage}\hfil
    \begin{minipage}{.9\columnwidth}
      \centering
      \includegraphics[width=.8\columnwidth]{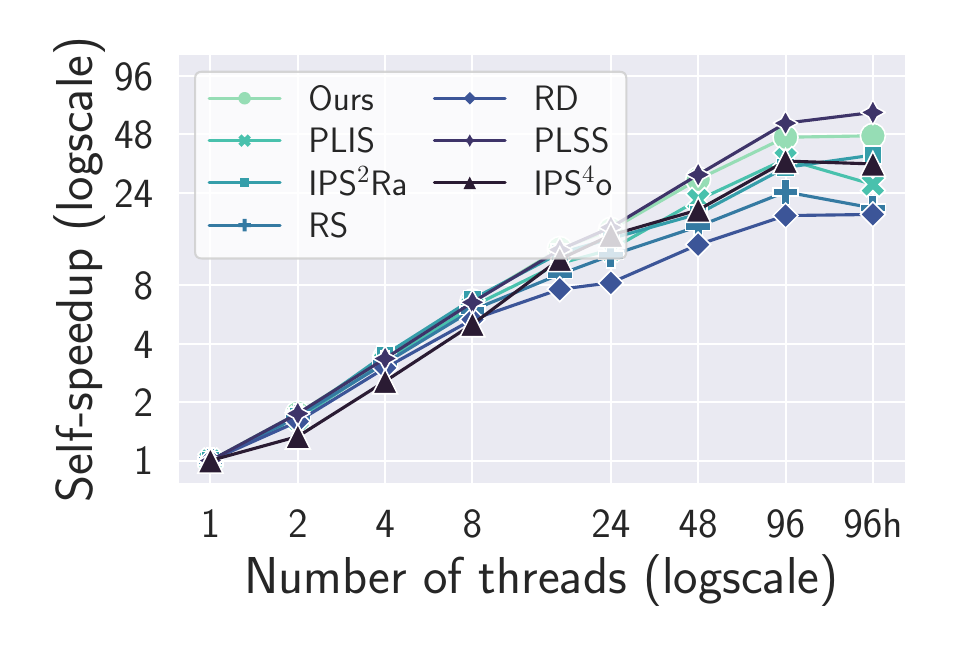}
  \vspace{-.5em}
      \caption{\textbf{Self-speedup with varying thread counts of all tested implementations on \exponential-$\boldsymbol{7}$.}}\label{fig:scalability-exponential-7e-5-64}
    \end{minipage}

\bigskip
    Zipfian Distribution\\
    \begin{minipage}{.9\columnwidth}
      \centering
      \includegraphics[width=.8\columnwidth]{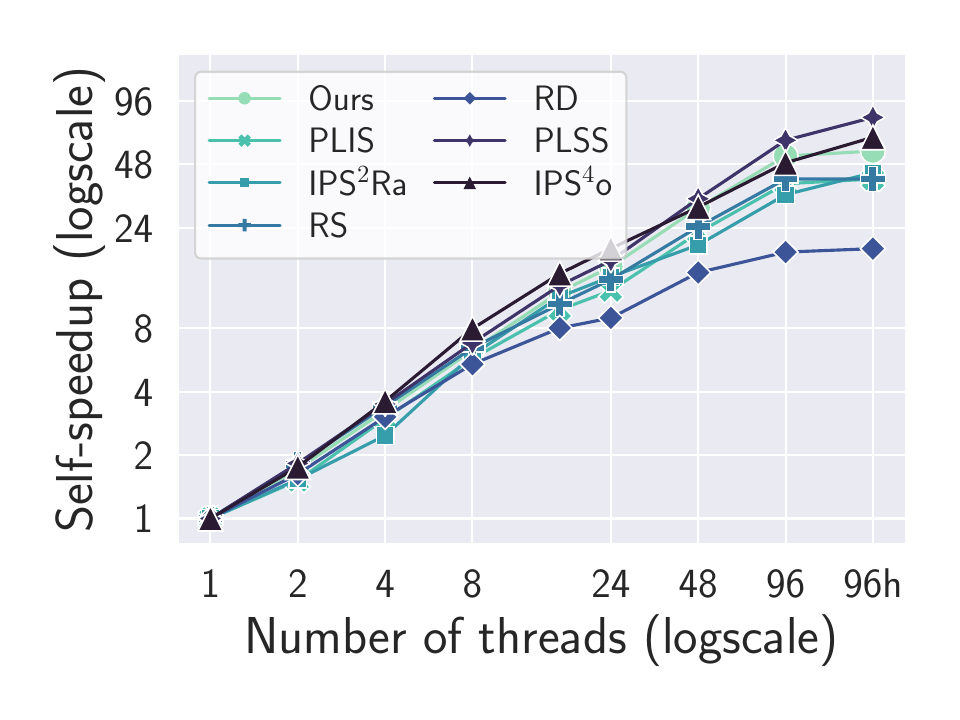}
  \vspace{-.5em}
      \caption{\textbf{Self-speedup with varying thread counts of all tested implementations on \zipfian-$\boldsymbol{0.8}$.}}\label{fig:scalability-Zipfian-0.7-64}
    \end{minipage}\hfil
    \begin{minipage}{.9\columnwidth}
      \centering
      \includegraphics[width=.8\columnwidth]{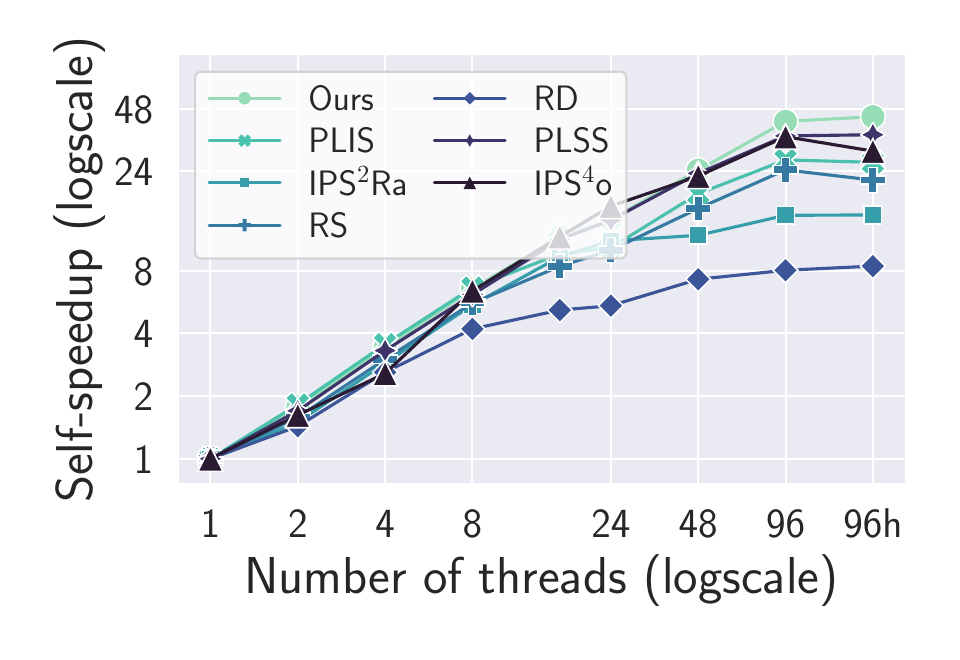}
  \vspace{-.5em}
      \caption{\textbf{Self-speedup with varying thread counts of all tested implementations on \zipfian-$\boldsymbol{1.2}$.}}\label{fig:scalability-Zipfian-1.2-64}
    \end{minipage}

\bigskip
    \bexp{} Distribution\\
    \begin{minipage}{.9\columnwidth}
      \centering
      \includegraphics[width=.8\columnwidth]{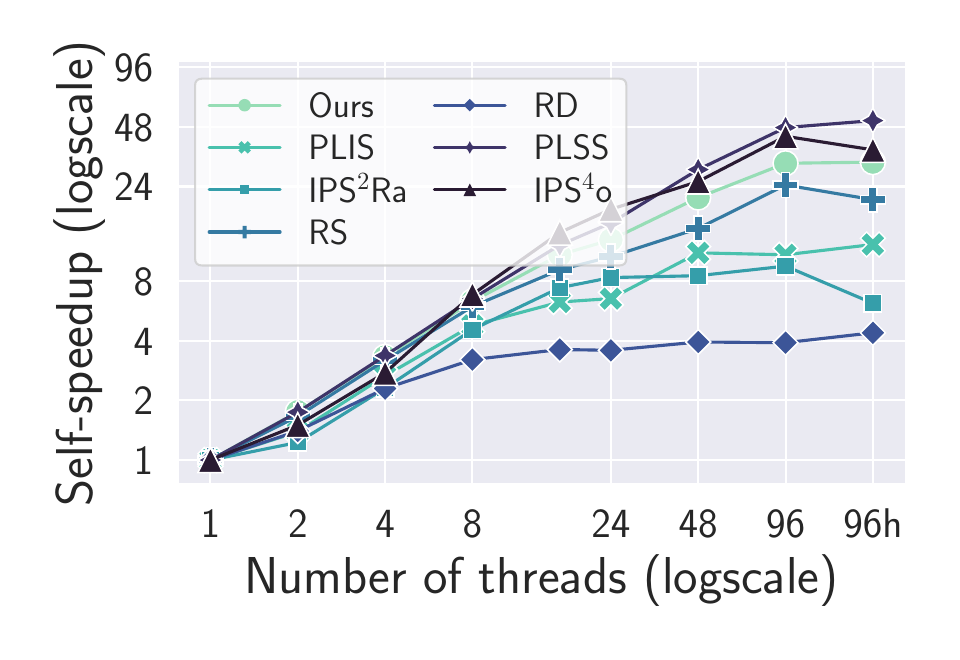}
  \vspace{-.5em}
      \caption{\textbf{Self-speedup with varying thread counts of all tested implementations on \bitexp-$\boldsymbol{30}$.}}\label{fig:scalability-bitsexp-30-64}
    \end{minipage}\hfil
    \begin{minipage}{.9\columnwidth}
      \centering
      \includegraphics[width=.8\columnwidth]{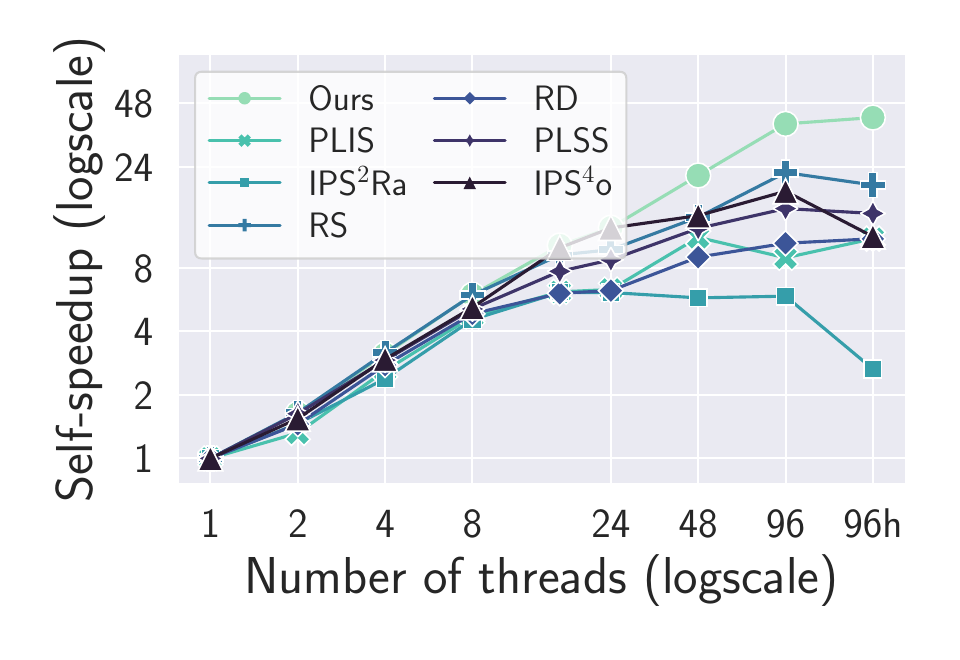}
  \vspace{-.5em}
      \caption{\textbf{Self-speedup with varying thread counts of all tested implementations on \bitexp-$\boldsymbol{100}$.}}\label{fig:scalability-bitsexp-100-64}
    \end{minipage}
    \end{figure*}

\begin{figure*}[t]
    {\Large Scalability with increasing input size (32-bit keys). Lower is better.}\\~\\
    Uniform Distribution\\
    \begin{minipage}{.9\columnwidth}
      \centering
      \includegraphics[width=.8\columnwidth]{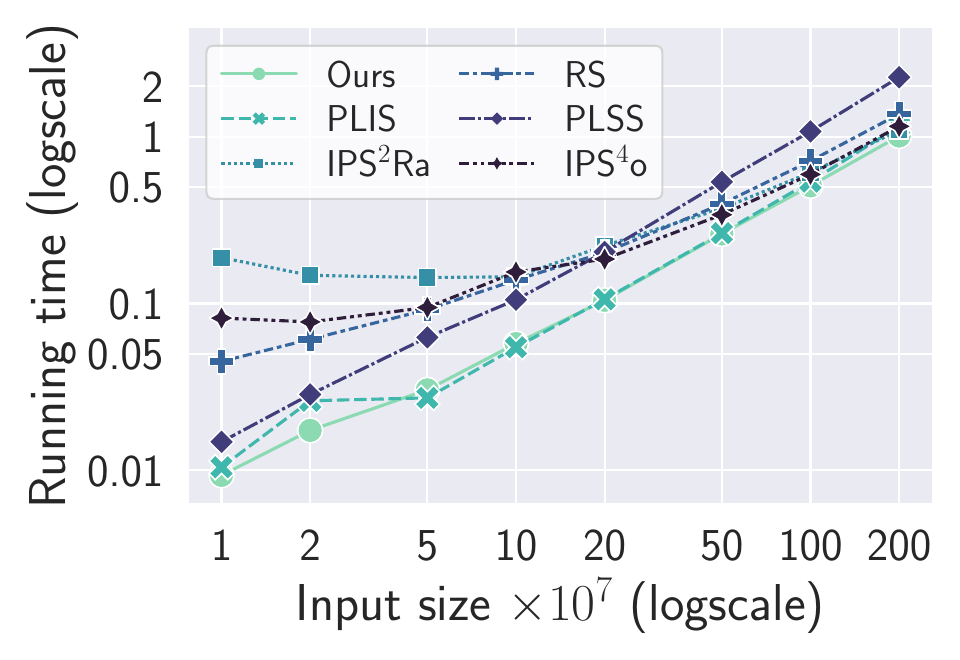}
      \vspace{-.5em}\caption{\textbf{Scalability with increasing input size ($n$) of all tested implementations on \uniform-$\boldsymbol{10^7}$.}}\label{fig:input-size-uniform-1e7-32}
    \end{minipage}\hfill
    \begin{minipage}{.9\columnwidth}
      \centering
      \includegraphics[width=.8\columnwidth]{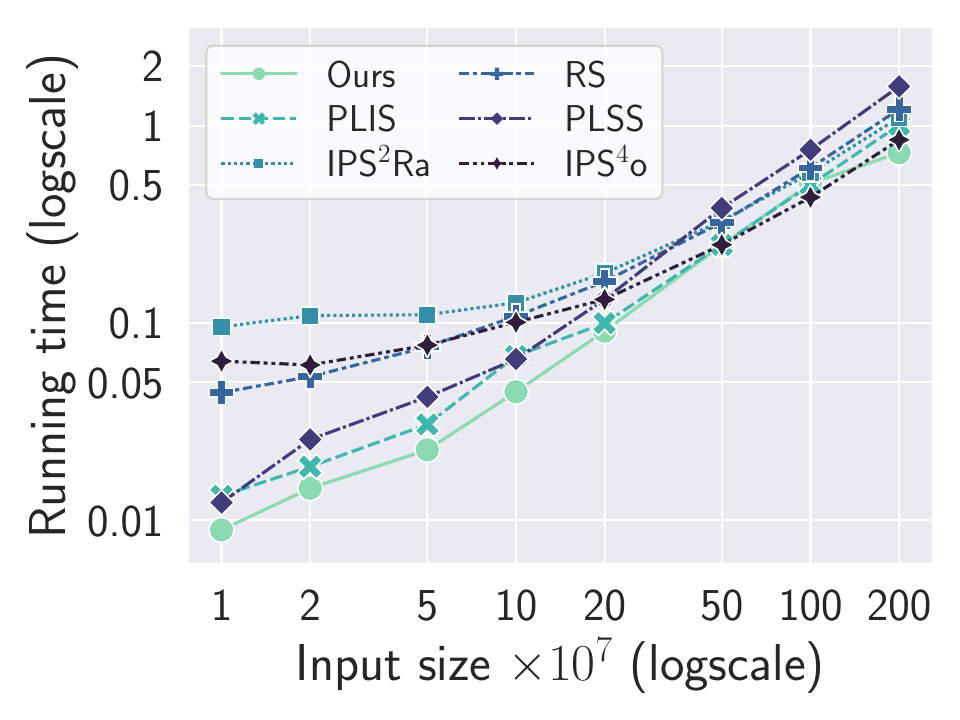}
      \vspace{-.5em}\caption{\textbf{Scalability with increasing input size ($n$) of all tested implementations on \uniform-$\boldsymbol{10^3}$.}}\label{fig:input-size-uniform-1e3-32}
    \end{minipage}\\
\bigskip
    Exponential Distribution\\
    \begin{minipage}{.9\columnwidth}
      \centering
      \includegraphics[width=.8\columnwidth]{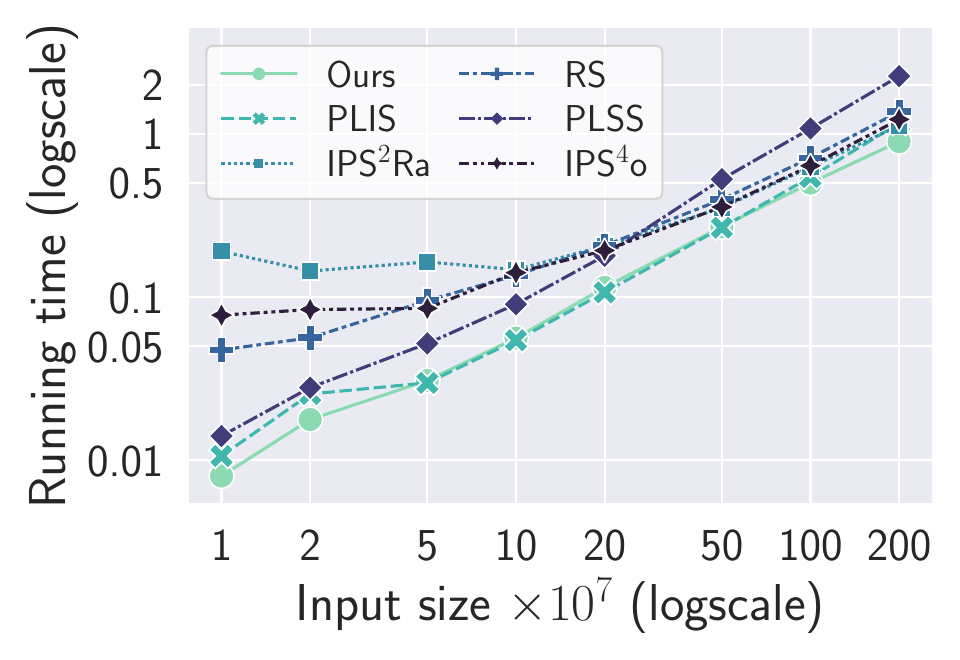}
      \vspace{-.5em}\caption{\textbf{Scalability with increasing input size ($n$) of all tested implementations on \exponential-$\boldsymbol{2}$.}}\label{fig:input-size-exponential-2e-5-32}
    \end{minipage}\hfill
    \begin{minipage}{.9\columnwidth}
      \centering
      \includegraphics[width=.8\columnwidth]{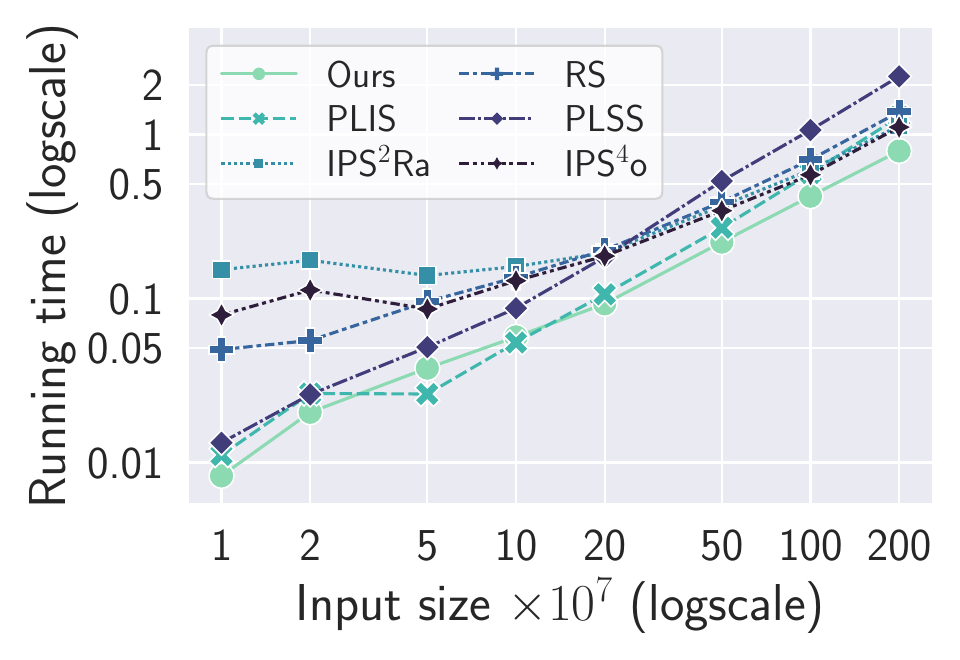}
      \vspace{-.5em}\caption{\textbf{Scalability with increasing input size ($n$) of all tested implementations on \exponential-$\boldsymbol{7}$.}}\label{fig:input-size-exponential-7e-5-32}
    \end{minipage}

\bigskip
    Zipfian Distribution\\
    \begin{minipage}{.9\columnwidth}
      \centering
      \includegraphics[width=.8\columnwidth]{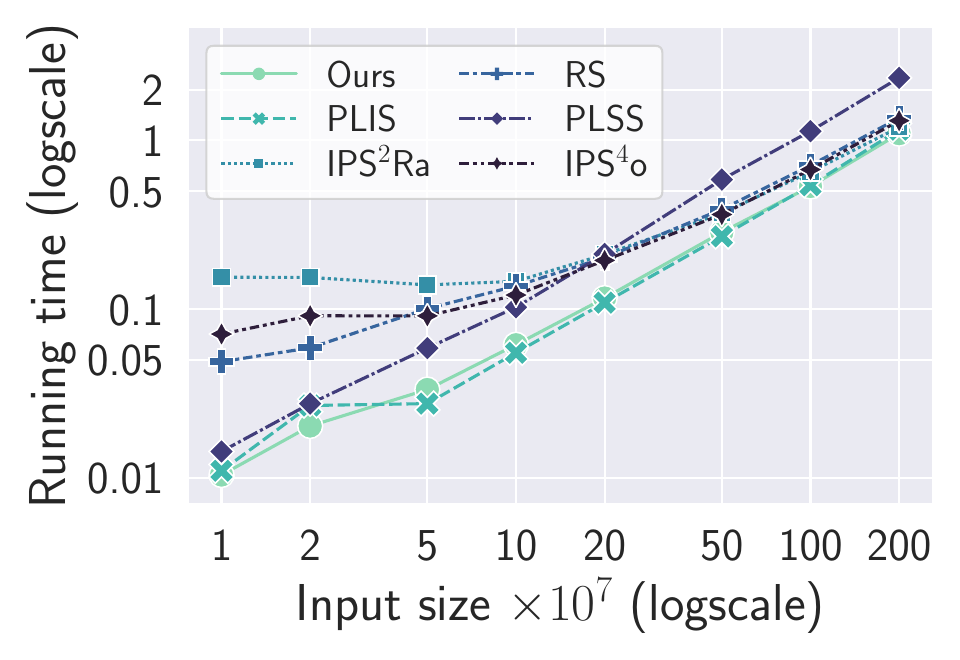}
      \vspace{-.5em}\caption{\textbf{Scalability with increasing input size ($n$) of all tested implementations on \zipfian-$\boldsymbol{0.8}$.}}\label{fig:input-size-Zipfian-0.8-32}
    \end{minipage}\hfill
    \begin{minipage}{.9\columnwidth}
      \centering
      \includegraphics[width=.8\columnwidth]{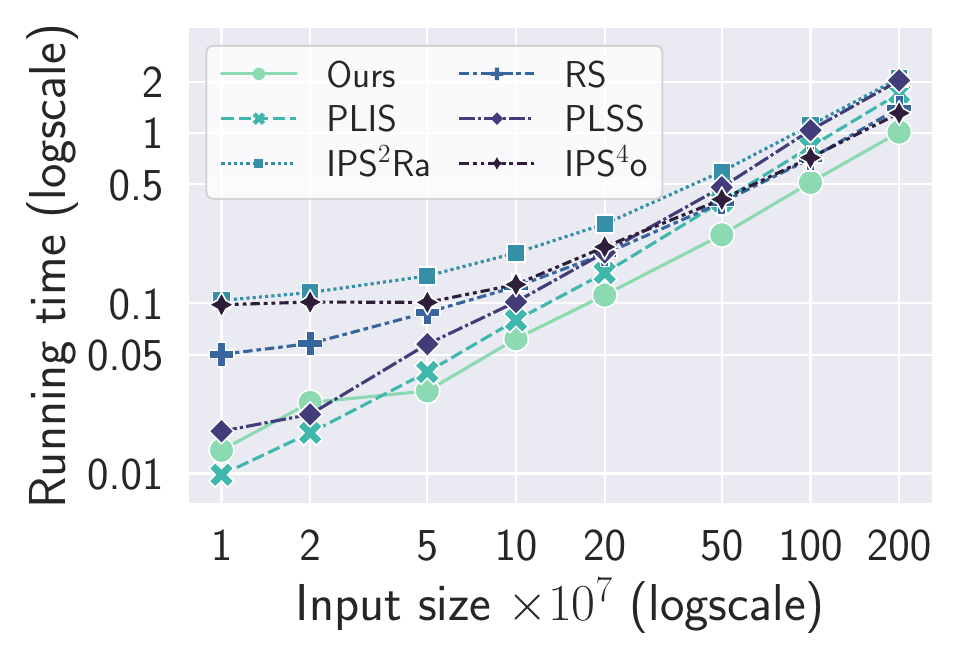}
      \vspace{-.5em}\caption{\textbf{Scalability with increasing input size ($n$) of all tested implementations on \zipfian-$\boldsymbol{1.2}$.}}\label{fig:input-size-Zipfian-1.2-32}
    \end{minipage}

\bigskip
    \bexp{} Distribution\\
    \begin{minipage}{.9\columnwidth}
      \centering
      \includegraphics[width=.8\columnwidth]{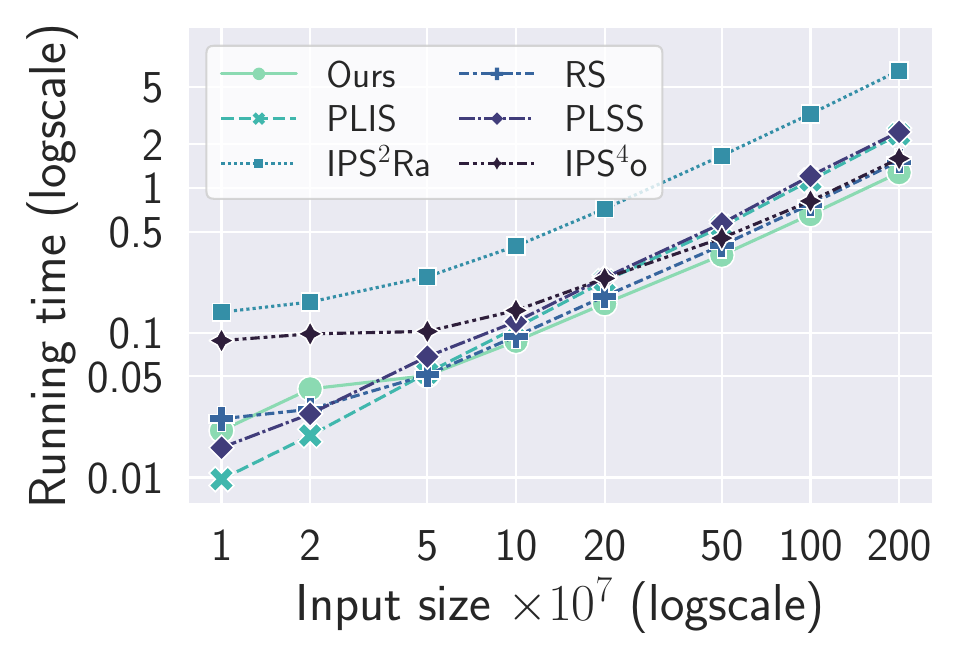}
      \vspace{-.5em}\caption{\textbf{Scalability with increasing input size ($n$) of all tested implementations on \bitexp-$\boldsymbol{30}$.}}\label{fig:input-size-bitsexp-30-32}
    \end{minipage}\hfill
    \begin{minipage}{.9\columnwidth}
      \centering
      \includegraphics[width=.8\columnwidth]{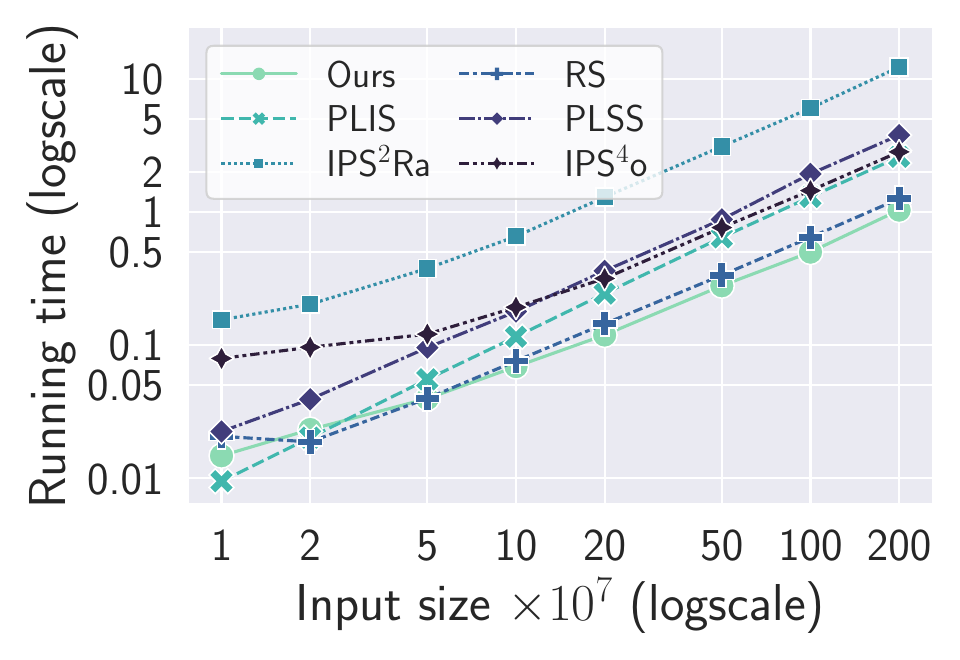}
      \vspace{-.5em}\caption{\textbf{Scalability with increasing input size ($n$) of all tested implementations on \bitexp-$\boldsymbol{100}$.}}\label{fig:input-size-bitsexp-100-32}
    \end{minipage}
    \end{figure*}

    \begin{figure*}[t]
        {\Large Scalability with increasing input size (64-bit keys). Lower is better.}\\~\\
        Uniform Distribution\\
        \begin{minipage}{.9\columnwidth}
          \centering
          \includegraphics[width=.8\columnwidth]{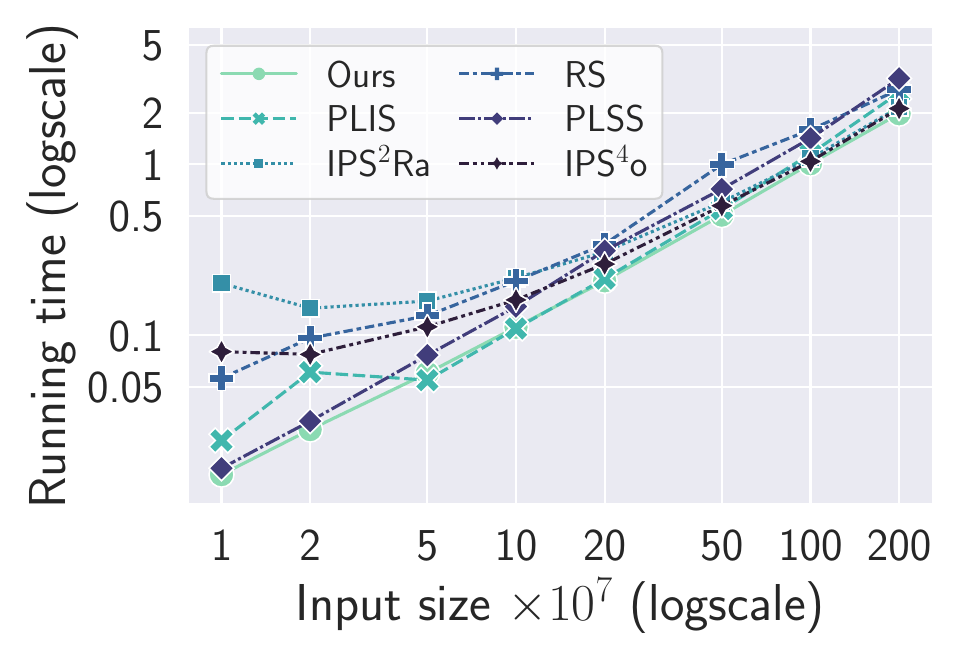}
          \vspace{-.5em}\caption{\textbf{Scalability with increasing input size ($n$) of all tested implementations on \uniform-$\boldsymbol{10^7}$.}}\label{fig:input-size-uniform-1e7-64}
        \end{minipage}\hfill
        \begin{minipage}{.9\columnwidth}
          \centering
          \includegraphics[width=.8\columnwidth]{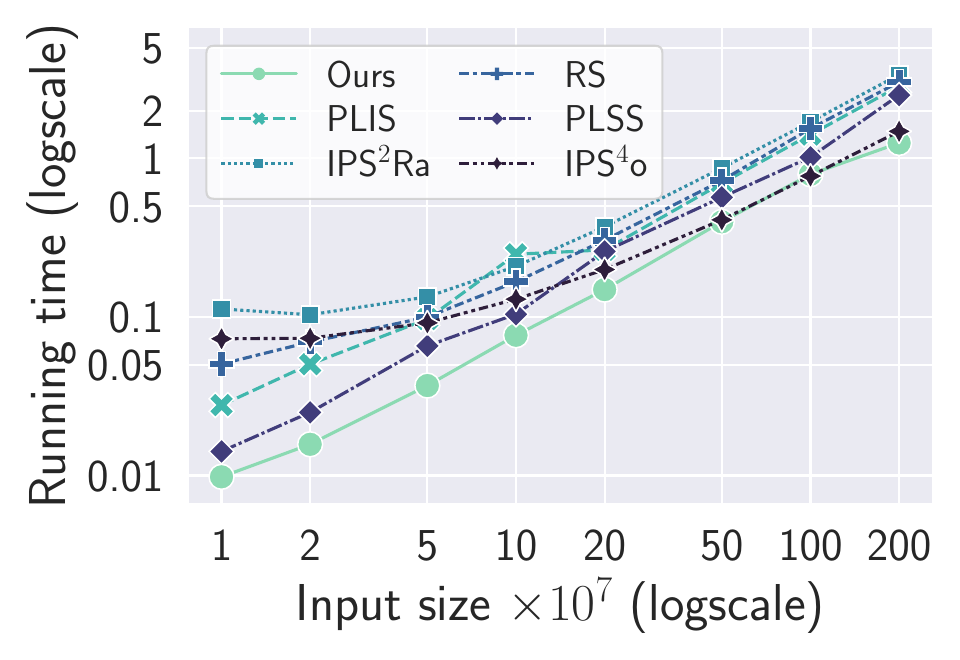}
          \vspace{-.5em}\caption{\textbf{Scalability with increasing input size ($n$) of all tested implementations on \uniform-$\boldsymbol{10^3}$.}}\label{fig:input-size-uniform-1e3-64}
        \end{minipage}

\bigskip
        Exponential Distribution\\
        \begin{minipage}{.9\columnwidth}
          \centering
          \includegraphics[width=.8\columnwidth]{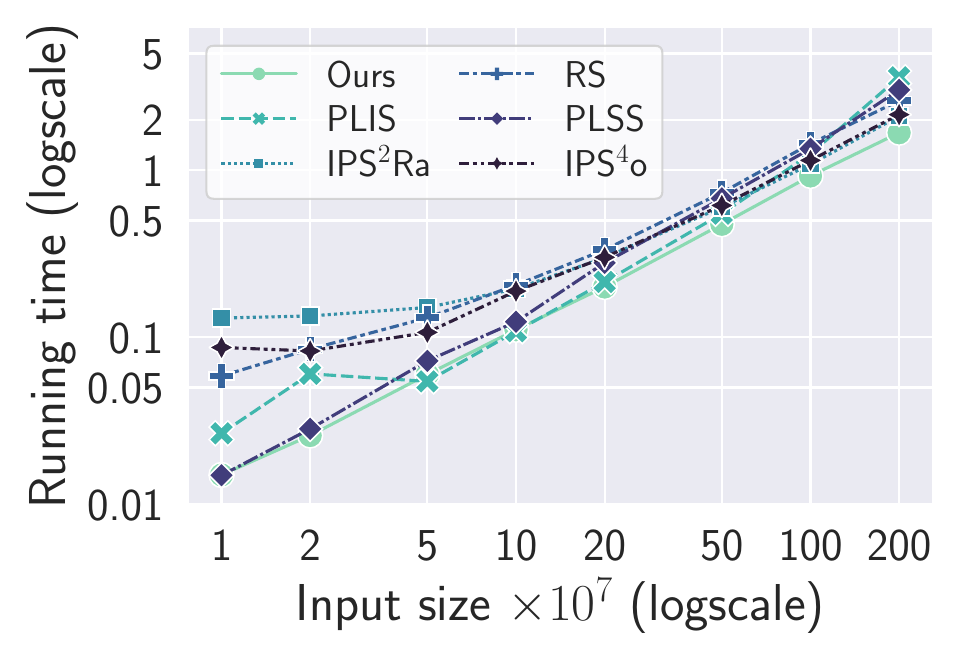}
          \vspace{-.5em}\caption{\textbf{Scalability with increasing input size ($n$) of all tested implementations on \exponential-$\boldsymbol{2}$.}}\label{fig:input-size-exponential-2e-5-64}
        \end{minipage}\hfill
        \begin{minipage}{.9\columnwidth}
          \centering
          \includegraphics[width=.8\columnwidth]{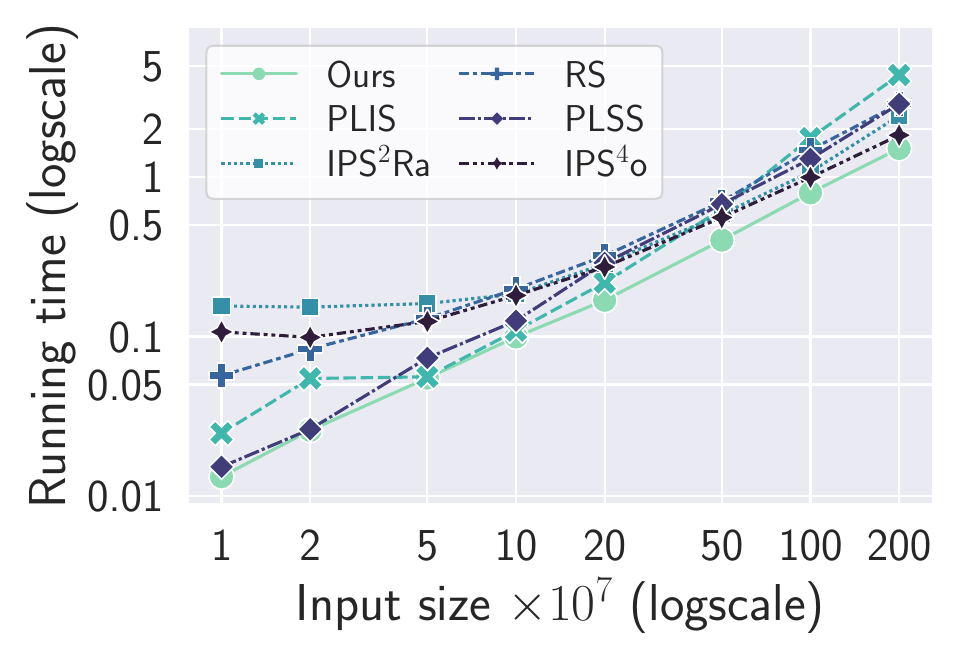}
          \vspace{-.5em}\caption{\textbf{Scalability with increasing input size ($n$) of all tested implementations on \exponential-$\boldsymbol{7}$.}}\label{fig:input-size-exponential-7e-5-64}
        \end{minipage}

\bigskip
        Zipfian Distribution\\
        \begin{minipage}{.9\columnwidth}
          \centering
          \includegraphics[width=.8\columnwidth]{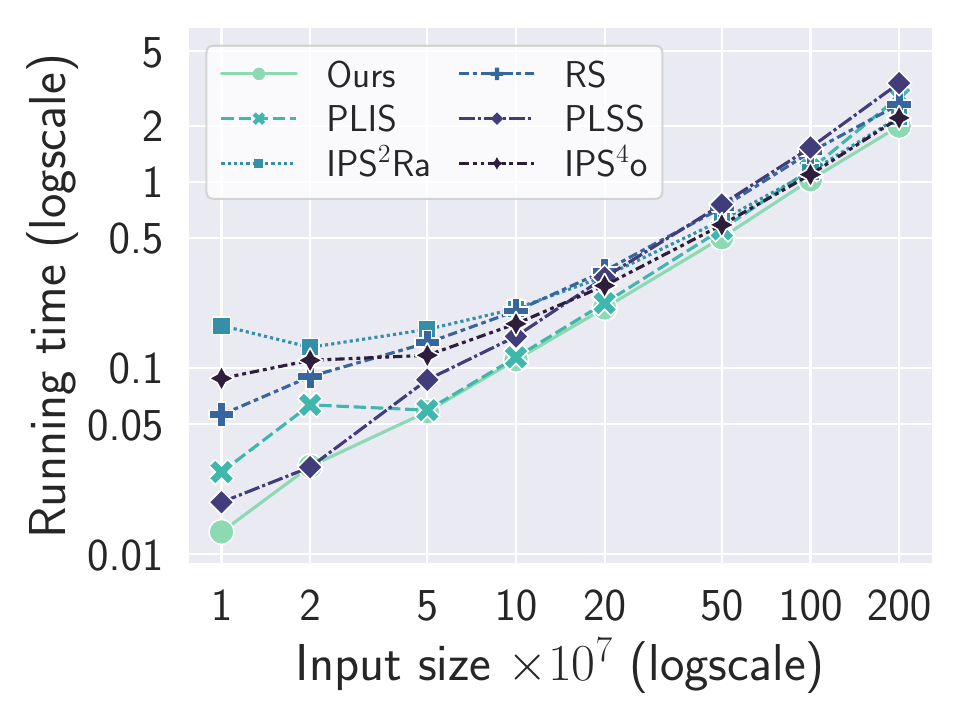}
          \vspace{-.5em}\caption{\textbf{Scalability with increasing input size ($n$) of all tested implementations on \zipfian-$\boldsymbol{0.8}$.}}\label{fig:input-size-Zipfian-0.8-64}
        \end{minipage}\hfill
        \begin{minipage}{.9\columnwidth}
          \centering
          \includegraphics[width=.8\columnwidth]{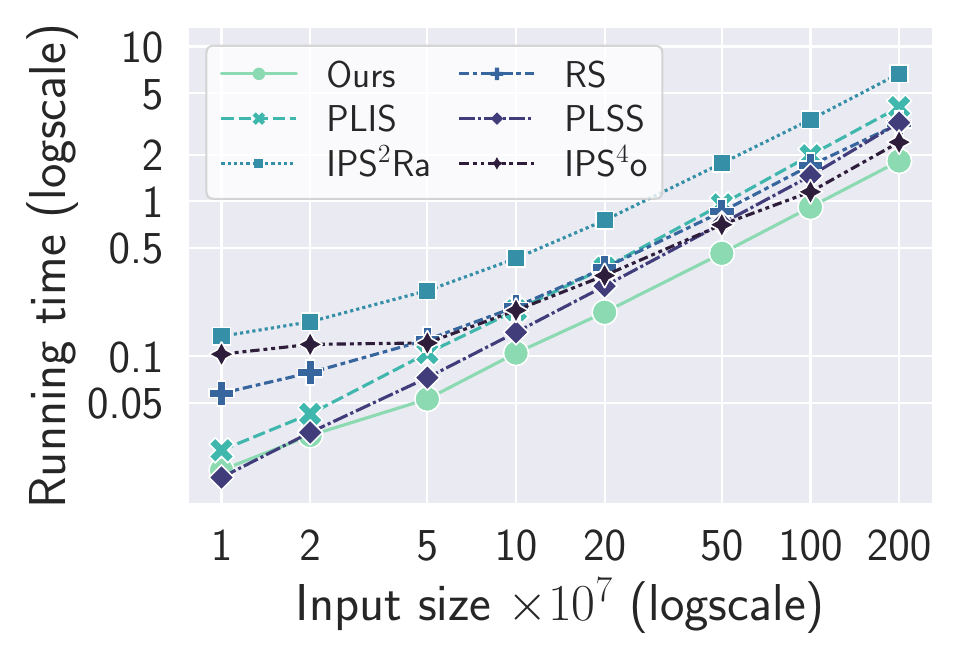}
          \vspace{-.5em}\caption{\textbf{Scalability with increasing input size ($n$) of all tested implementations on \zipfian-$\boldsymbol{1.2}$.}}\label{fig:input-size-Zipfian-1.2-64}
        \end{minipage}

\bigskip
        \bexp{} Distribution\\
        \begin{minipage}{.9\columnwidth}
          \centering
          \includegraphics[width=.8\columnwidth]{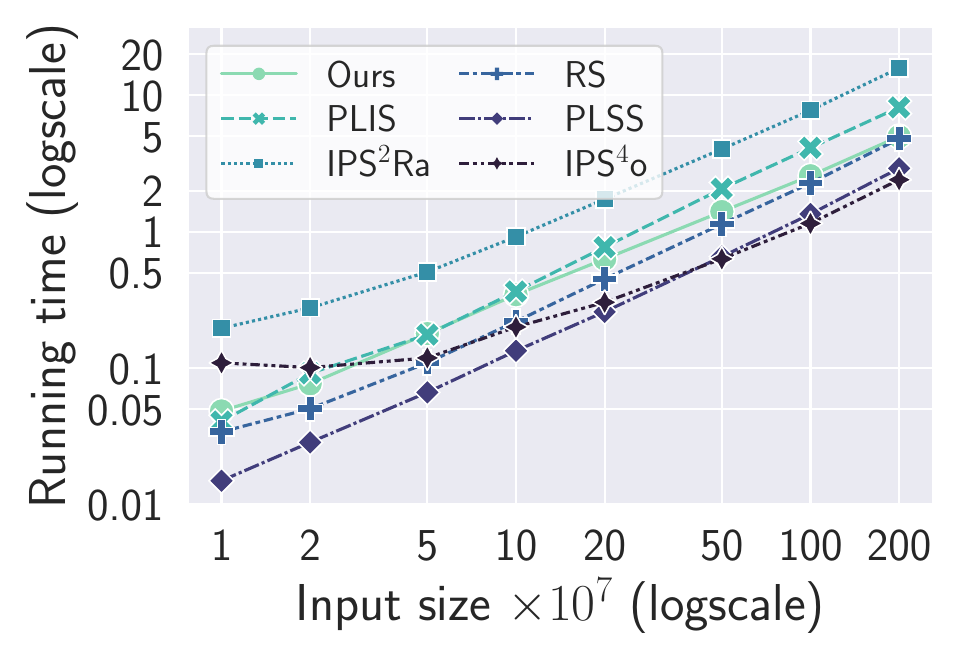}
          \vspace{-.5em}\caption{\textbf{Scalability with increasing input size ($n$) of all tested implementations on \bexp-$\boldsymbol{30}$.}}\label{fig:input-size-bitsexp-30-64}
        \end{minipage}\hfill
        \begin{minipage}{.9\columnwidth}
          \centering
          \includegraphics[width=.8\columnwidth]{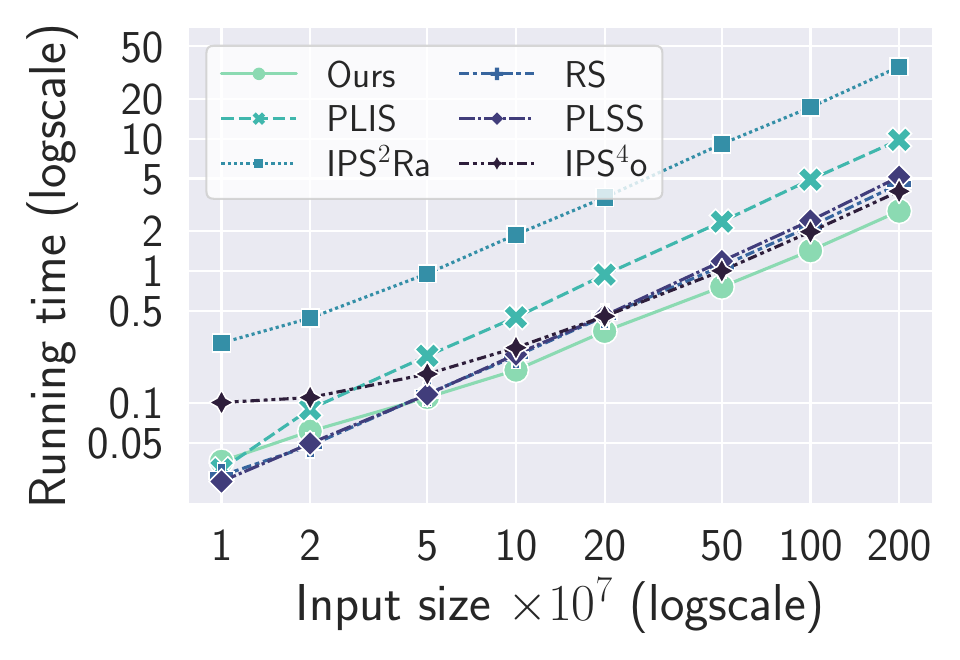}
          \vspace{-.5em}\caption{\textbf{Scalability with increasing input size ($n$) of all tested implementations on \bexp-$\boldsymbol{100}$.}}\label{fig:input-size-bitsexp-100-64}
        \end{minipage}
        \end{figure*}

}

\end{document}
\endinput